\documentclass{article}

\usepackage{cite}
\usepackage{algorithmic}
\usepackage{graphicx}
\usepackage{algorithm,algorithmic}
\usepackage{hyperref}
\usepackage{textcomp}
\usepackage{custom}
\usepackage[margin=1in]{geometry}
\usepackage{multirow}
\usepackage{booktabs} 

\usepackage{subfigure}
\usepackage{float}
\usepackage{todonotes}
\usepackage{xcolor}

\usepackage{soul}

\begin{document}
\title{Fast Distributed Optimization over Directed Graphs under Malicious Attacks using Trust}
\author{Arif Kerem Day\i\thanks{Co-primary Authors}, Orhan Eren Akg\"un\footnotemark[1], Stephanie Gil, Michal Yemini, Angelia Nedi\'c
\thanks{A.~K.~Day\i, O.~E.~Akg\"un, and S.~Gil are with the School of Engineering and Applied Sciences, Harvard University, USA. M.~Yemini is with the Faculty of Engineering, Bar-Ilan University, Israel. A.~Nedi\'c is with the School of Electrical, Computer and Energy Engineering, Arizona State University, USA.
}
}
\date{}
\maketitle

\begin{abstract}
In this work, we introduce the Resilient Projected Push-Pull (RP3) algorithm designed for distributed optimization in multi-agent cyber-physical systems with directed communication graphs and the presence of malicious agents. Our algorithm leverages stochastic inter-agent trust values and gradient tracking to achieve geometric convergence rates in expectation even in adversarial environments. We introduce growing constraint sets to limit the impact of the malicious agents without compromising the geometric convergence rate of the algorithm. We prove that RP3 converges to the nominal optimal solution almost surely and in the $r$-th mean for any $r\geq 1$, provided the step sizes are sufficiently small and the constraint sets are appropriately chosen. We validate our approach with numerical studies on average consensus and multi-robot target tracking problems, demonstrating that RP3 effectively mitigates the impact of malicious agents and achieves the desired geometric convergence.
\end{abstract}


\section{Introduction}
\label{sec:introduction}
In this work, we are interested in distributed optimization problems involving minimizing the sum of agents' individual strongly convex loss functions, potentially over closed and convex constraint sets in the presence of malicious agents.
Distributed optimization lays the foundation for many algorithms in multi-robot systems \cite{schwager_tutorial_2024} and sensor networks \cite{nedic_control_2018} such as collaborative manipulation \cite{schwager_collaborative_manipulation}, distributed control \cite{nedic_control_2018, christofides_MPC_2013}, localization \cite{roberto_localization_2014, Long_localization_2016}, and estimation \cite{nedic_distributed_sensors,ola_tracking_2020}. In this work, we focus on two important challenges in the analysis of such distributed optimization problems and their compounding impact: 1) having a directed communication graph, and 2) the presence of malicious agents in the system. The study of directed communication graphs is crucial for applications where the communication capabilities of the agents are heterogeneous. However, having asymmetric information flow requires careful design of the distributed optimization algorithms \cite{nedic2014distributed, makhdoumi2015graph,xi2017distributed}. The other great challenge is the presence of malicious agents, which may cause catastrophic effects on the performance of multi-agent systems when there are no precautions taken \cite{sundaram2010distributed,sundaram2018distributed,ourTRO,yemini2022resilience,cavorsi_rht_2023}. Despite considerable advancements in distributed optimization research, the combined impact of directed communication graphs and malicious agents remains an under-explored area. Our goal in this paper is to develop and analyze a fast resilient distributed optimization algorithm for directed communication graphs. Specifically, we want the algorithm to have geometric convergence in directed graphs even in the presence of malicious agents.

To mitigate the impact of malicious agents, a growing body of literature investigates the use of the physical channels of information in cyber-physical systems. Agents can use many methods such as camera inputs, sensor observations, and wireless fingerprints to derive inter-agent trust values and assess the trustworthiness of their neighbors \cite{gil2017guaranteeing,Pierson2016,xiong2023securearray} (see \cite{gil2023physicality} for a survey). These trust values are stochastic and imperfect as they come from noisy physical information \cite{gil2017guaranteeing}. However, agents can accumulate more values over time to have a better estimate of trustworthiness of the agents they interact with. It has been shown that these inter-agent trust values in such cases can lead to strong theoretical guarantees for multi-agent systems \cite{yemini2022resilience, cavorsi_rht_2023, ourTRO, aydın2024multiagent, hadjicostis2023trustworthy}. Moreover, methods that exploit inter-agent trust values do not require additional assumptions limiting the number of tolerable malicious agents, or their strategies \cite{yemini2022resilience, ourTRO, hadjicostis2023trustworthy}, unlike methods that rely solely on transmitted data to eliminate malicious information \cite{sundaram2010distributed, sundaram2018distributed, ravi_2019}. Recent work \cite{yemini2022resilience} shows that it is possible to retain the global optimal value in distributed optimization by leveraging the inter-agent trust values. However, the existing results are limited in three ways: 1)~the results are only applicable to undirected graphs, 
2)~the algorithm only works for constrained optimization problems with a compact constraint set, and 3)~the algorithm requires decreasing step-sizes, resulting in slow convergence. On the other hand, and for the case without malicious agents, it is possible to obtain a geometric convergence rate with fixed step sizes over directed graphs using gradient tracking methods \cite{AB_usman, push-pull, xin2019frost}. In these methods, agents store an additional variable, called the gradient tracking variable, to estimate the global gradient and speed up the convergence \cite{harnessing_smoothness_journal}. However, these algorithms work under the assumption that all the agents are fully cooperative, and therefore, cannot handle malicious agents. Our aim in this work is to develop a more general distributed optimization algorithm that achieves geometric convergence rate with a fixed step size in the directed graph case with malicious agents.

In our previous work \cite{L4DC,L4DC_2023_extended}, we developed a learning protocol that enables agents to develop opinions of trust about both their in and out neighbors leveraging trust values. In this work, we leverage this learning protocol to develop a resilient distributed optimization algorithm, referred to as Resilient Projected Push-Pull (RP3). The RP3 algorithm uses gradient tracking to achieve geometric convergence over directed graphs. In the RP3, agents use trust opinions to form a trusted neighborhood and consider only the agents in this neighborhood when performing their updates. Integrating the learning protocol with gradient tracking presents two primary challenges. First, agents' trust opinions improve over time, which initially allows malicious agents opportunities to influence their neighbors. Moreover, at any time, legitimate agents do not know if their opinions are perfectly accurate or not. Second, the gradient tracking method introduces new attack surfaces as agents need to share gradient tracking variables with each other, which can be influenced by the malicious agents arbitrarily. To limit the malicious influence until agents' trust opinions improve, we introduce growing constraint sets that agents project both their decision and gradient tracking variables onto. We demonstrate that appropriately chosen constraint sets can restore the algorithm's nominal performance, still achieving geometric convergence while containing the influence of malicious agents. These sets also allow us to extend our results to the optimization problems over unbounded constraint sets. Our contributions in this paper are as follows
\begin{enumerate}
	\item We introduce the Resilient Projected Push-Pull algorithm for constrained distributed optimization problems over directed graphs with malicious agents. We show that the algorithm converges to the nominal optimal solution almost surely and in
 the $r$-th mean for any $r\geq 1$.
	\item We demonstrate that, with sufficiently small step sizes and appropriately chosen constraint sets (as characterized in this paper), the algorithm achieves a geometric convergence rate in expectation.
\end{enumerate}
Finally, we apply our algorithm to average consensus and multi-robot target tracking problems in numerical studies to validate our theoretical results.

\section{Related Works}
Achieving fast convergence in distributed optimization problem with fully cooperative agents is well studied in the literature. The gradient tracking method was introduced in \cite{harnessing_smoothness_cdc, xu2015augmented, DIGing} to achieve a geometric convergence rate over unconstrained distributed optimization problems over undirected graphs. The algorithm was extended to directed graphs in \cite{DIGing} by employing column stochastic mixing matrices and in \cite{FROST_og} by using row stochastic mixing matrices. Both methods require distributed estimation of the non-one Perron vector of the mixing matrix. The Push-Pull algorithm introduced in \cite{AB_usman, push-pull} utilizes both row and column stochastic matrices to attain geometric convergence in directed graphs without needing to estimate the non-one Perron vector. Further developments include extensions of row stochastic mixing matrix-based gradient tracking methods to constrained optimization problems in \cite{liu2020discrete,luan2023distributed,scutari2019distributed}. Moreover, \cite{projected_push_pull} extends the Push-Pull algorithm to constrained optimization problems. Our method is closely related to the \cite{projected_push_pull} algorithm, chosen due to the sensitivity of methods in \cite{DIGing,FROST_og,liu2020discrete,luan2023distributed,scutari2019distributed} to the initialization of the variables. Such dependency on initialization is undesirable in the presence of malicious agents. Notably, none of these studies account for malicious agents; they assume that all agents are fully cooperative, and share accurate information. Studies on noisy information sharing in gradient tracking \cite{robust2020,basar_noise_GT_2023,nedic_DP_GT_2024,Zhaoye_noise_GT_2024, wenwen_noise_GT_2023,wu_noise_GT_2022} have relaxed the assumption that the information sharing between the agents is perfect. These works and works that consider stochastic gradient information \cite{nedic_DSGT_2021,xin_CDC_SGT_2019,zhao_DSGT_2024} typically assume that the noise affecting the system is independent and unbiased. In contrast, our research considers adversarial inputs that can be arbitrary and strategically chosen by adversaries. Additionally, we explore the dynamics where agents accumulate and update trust values over time. This process introduces correlated noise across the agents' weights, making the previous analysis that assumes independent noise inapplicable to our problem.

The impact of malicious agents differ significantly from unbiased statistical noise. Previous results in \cite{sundaram2018distributed} show that even a single malicious agent can force consensus based distributed optimization algorithms to converge to an arbitrary value. To mitigate the impact of the malicious agents, a variety of resilient distributed optimization algorithms have been developed  \cite{sundaram2018distributed,Su2021ByzantineResilientMO,Kuwaranancharoen2024ScalableDO,Zhang2024AcceleratedDO,Gupta2021ByzantineFI,Zhu2022ResilientDO,Gupta2020ResilienceIC,Ravi2019DetectionAI,ravi_2019,Zhao2020ResilientDO,Xu2022ATR,Fu2020ResilientCD,Kaheni2022ResilientCO}. These methods often utilize filtering based on the values received from other agents to achieve resilience \cite{sundaram2018distributed,Su2021ByzantineResilientMO,Gupta2021ByzantineFI,Zhu2022ResilientDO,Gupta2020ResilienceIC,Fu2020ResilientCD}. Later on, these filtering-based approaches are extended to multidimensional functions in \cite{Kuwaranancharoen2024ScalableDO}, and constrained optimization problems in \cite{Kaheni2022ResilientCO}. The work in \cite{Zhang2024AcceleratedDO} introduces gradient tracking to these methods but still requires a decreasing step size. Common limitations of these approaches include restrictions on the number of tolerable malicious agents and the network topology, and convergence typically only to the convex hull of the agents' local minimizers, rather than to a true optimal point. Works  \cite{Gupta2021ByzantineFI,Zhu2022ResilientDO,Gupta2020ResilienceIC,Kaheni2022ResilientCO} demonstrate that exact convergence to the true optima is possible with redundancy in the cost functions of legitimate agents. Techniques involving agents estimating and cross-checking their neighbors' gradients to perform detection and filtering are discussed in \cite{Ravi2019DetectionAI,ravi_2019}. However, these methods only consider specific attack types and do not guarantee convergence to the optimal point. Studies \cite{Zhao2020ResilientDO,Xu2022ATR} propose the use of trusted agents that are known to everyone to overcome limitations regarding the number of tolerable malicious agents. However, these methods introduce additional assumptions about connectivity between trusted and regular nodes, and guarantee only convergence to the convex hull of the trusted agents' local minimizers. In our work, we guarantee geometrically fast convergence to a \emph{true optimal point} without any limitation on the number of tolerable malicious agents. Moreover, our algorithm works for both constrained and unconstrained optimization problems, and in directed graphs. In contrast to some of the existing work, we do not require a set of pre-existing trusted agents or redundancy in the cost functions. 

In our approach, we use inter-agent trust values that can be derived from physical properties of the cyberphysical systems as explored in previous research \cite{gil2017guaranteeing,cavorsi2023ICRA,Pierson2016,xiong2023securearray,cheng2021general,pippin2014trust,yang2024enhancing}. A comprehensive survey of these inter-agent trust values is available in \cite{gil2023physicality}. Examples of such trust values include observations of other robots \cite{yang2024enhancing, cheng2021general} or vehicles \cite{pippin2014trust}, and using wireless finger profiles from incoming transmissions \cite{gil2017guaranteeing,cavorsi2023ICRA,xiong2023securearray}. Furthermore, the study in \cite{ourTRO} demonstrated that since these approaches leverage physical information independent of the data transmitted by the agents to assess trustworthiness, it is possible to achieve resilience even when a majority of the agents in the network are malicious in undirected graphs. In these works, agents can develop more accurate trust estimations over time by aggregating more observations about their neighbors. The learning protocol introduced in \cite{L4DC_2023_extended, L4DC} further enhances this by enabling agents to learn about the trustworthiness of the entire network through propagated trust opinions. In our work, we leverage this protocol to enable agents to develop trust estimations about their in and out neighbors in directed graphs. 

The paper~\cite{yemini2022resilience} has developed a resilient distributed optimization utilizing inter-agent trust values for constrained problems with compact and convex constraint sets in undirected graphs. Our method diverges from \cite{yemini2022resilience} by employing gradient tracking to achieve a geometric convergence rate in expectation, rather than the slower rate from the decreasing step size used in \cite{yemini2022resilience}. Moreover, as opposed to \cite{yemini2022resilience} where malicious agents can only influence the decision variables, malicious agents can manipulate both the decision variables and gradient tracking variables in our setup. For example, \cite{Ding_GT_attack_CDC_2020} proposes an attack model where malicious agents can manipulate the gradient tracking variable to achieve convergence to an arbitrary value. This creates additional challenges for achieving resilience in our setup and prevent us from relying on the analysis of \cite{yemini2022resilience}. Finally, unlike \cite{yemini2022resilience}, our algorithm works for directed graphs and for constrained optimization problems with closed and convex constraint sets. The closedness assumption is less restrictive than the compactness assumption in \cite{yemini2022resilience} and makes our method applicable to unconstrained optimization problems. To achieve this, we introduce a strategy of projecting decision and gradient tracking variables onto an increasing sequence of sets, mitigating the impact of malicious agents until more precise trust estimations are developed. Notably, to our knowledge, this is the first analysis that considers projection of the gradient tracking variables.

\section{Problem Formulation}
\subsection{Notation}
We use $\norm{\cdot}$ to denote the Euclidean norm. We define the $u$-weighted norm of $\boldx \in \R^d \times \dots \times \R^d$ ($n$ copies of $\R^d$) as $\norm{\boldx}_{u}=\sqrt{\sum_{i=1}^n u_i \norm{x_i}^2}$ where $x_i \in \R^d$ for any vector $u \in \R^n$ with $u_i > 0$ $\forall i$. We use $\E[Z]$ and $\E[Z|\mathcal{A}]$ to denote the expectation of a random variable $Z$ and the conditional expectation of $Z$ conditioning on the event $\mathcal{A}$, respectively. When $\mathcal{A}$ is empty, we define\footnote{This notation is needed for the cases we use the total law of probability in our analysis with the empty events. For the sake of simplicity of presentation, such an expectation will be used in the form $\E[Z|\mathcal{A}]\Pr(A)$ where $\Pr(\mathcal{A})$ equals $0$.} 
$\E[Z|\mathcal{A}]$ as $0.$ We will use the following definition of the growth of the set sequences in this work.

\begin{definition}[Growth of the set sequence $\{\mX_k \}$]
    For a non-empty set $\mX$, we define its size as $\norm{\mX} \triangleq \sup\{\|x\|: x\in\mX \}$. We let $\norm{\mX}=\infty$ if $\mX$ is unbounded. Moreover, when discussing the growth of a set sequence $\{\mX_k \}$, we are specifically referring to the growth of the sequence $\{\norm{\mX_k} \}$.
\end{definition}

  Finally, we define the projection operator as follows.
\begin{definition}[Projection onto $\mX$]
    Let $\mX\subseteq \R^d$ be nonempty, closed, and convex. Then, the projection operator $\projx{\cdot}: \R^d \to \R^d$ is defined as follows
    \begin{align*}
        \projx{x} = \arg\min_{z\in \mX} \norm{x-z}.
    \end{align*}
\end{definition}

\subsection{Problem Setup}
We consider a multi-agent system with a \textit{directed} communication graph $\mG=(\mV, \mE)$, where $\mV$ with $|\mV|=n$ denotes the set of agents and $\mE$ denotes the directed communication links. If agent $i$ can send information to agent $j$, then there is an edge $(i,j)\in \mE$ and we say that $j$ is an out-neighbor of $i$ and $i$ is an in-neighbor of $j$. We assume that every agent has a self-loop, i.e., $(i,i)\in \mE$ for all $i\in \mV$. We are interested in the case where an unknown subset of agents in the system are malicious We denote the set of malicious agents by $\mM \subset \mV$. Malicious agents are non-cooperative and can act arbitrarily. The set of cooperative agents is denoted by $\mL$ and referred to as legitimate agents. We assume that $\mL \cap \mM = \emptyset$ and $\mL \cup \mM = \mV$, i.e., an agent in the system is either legitimate or malicious. We denote the number of malicious agents by $n_{\mM}\triangleq |\mM|$ and the number of legitimate agents by  $n_{\mL}\triangleq |\mL|$. The sets $\mM$ and $\mL$ are defined for analytical purposes, and the legitimate agents do not know which agents in the system are legitimate or malicious. We say that malicious agents are untrustworthy and legitimate agents are trustworthy.

Each legitimate agent $i$ has a private local cost function, denoted by $f_i(x)$, that is only known to agent $i$. The legitimate agents' goal is to solve the following minimization problem, 
without revealing their private cost functions while exchanging information over the links $\mE$ of $\mG$, 
\begin{equation}
    \label{eq:objective_defn}
    \min_{x\in \mX\subseteq \R^d}f(x)\ \text{, where }\ f(x)\triangleq\frac{1}{n_\mL}\sum_{i\in \mL} f_i(x).
\end{equation}
The following assumption on the agents' cost functions is used.
\begin{assumption}\label{assumption:cost_function}
    For all legitimate agents $i\in \mL$, $f_i(x)$ is $\mu$-strongly convex, i.e, for some $\mu >0$, we have $\langle \nabla f_i(x)-\nabla f_i(y), x- y\rangle \geq \mu \normsq{x-y}$, for all $x,y \in \R^d.$ Moreover, for all legitimate agents $i\in \mL$, $\nabla f_i(x)$ is $L$-Lipschitz continuous, i.e, for some $L >0$, $\norm{\nabla f_i(x)-\nabla f_i(y)}\leq L \norm{x-y}$, for all $x,y \in \R^d.$ 
\end{assumption}
We consider two different assumptions on the constraint set $\mX$.
\begin{assumption}[Compact and Convex Constraint Set]\label{assumption:compact_convex_const_set}
    The constraint set $\mX \subseteq \R^d$ is nonempty, compact, and convex. Thus, there exist a scalar value $B>0$ such that
    \begin{equation}
        \norm{x}\leq B,\ \forall x \in \mX.
    \end{equation}
\end{assumption}
Having a compact set imposes a bound on the impact that malicious agents can have. For example, malicious agents cannot send values with $\norm{x}> B$ when \Cref{assumption:compact_convex_const_set} is known to be true since sending values with a norm $B$ would reveal their identities. To make our method applicable to a broader set of problems, we also consider a less restrictive assumption.
 \begin{assumption}[Closed and Convex Constraint Set]\label{assumption:closed_convex_const_set}
    The constraint set $\mX \subseteq \R^d$ is nonempty, closed, and convex.
\end{assumption}
\Cref{assumption:closed_convex_const_set} on the constraint set require different treatment than \Cref{assumption:compact_convex_const_set} in our analysis since there is no natural bound on the values malicious agents can send. In \Cref{sec:unbounded_case}, we will introduce a method to bound the effect of malicious agents when  \Cref{assumption:closed_convex_const_set} holds true. This way we generalize our results for more general constrained optimization problems, including unconstrained optimization problems. 
Note that under Assumption~\ref{assumption:cost_function} and either 
Assumption~\ref{assumption:compact_convex_const_set} or Assumption~\ref{assumption:closed_convex_const_set}, the problem in~\eqref{eq:objective_defn} has a unique solution, 
denoted by $x^*$.

In this work, we are interested in the problems where agents receive stochastic observations of trust from other agents that send information to them. In practice, this information can be obtained in various ways, including sensors onboard the agents, wireless fingerprints of the communication signals, and agent behaviors (see \cite{gil2023physicality} for a survey of such methods). In our previous work \cite{L4DC, L4DC_2023_extended}, we presented a learning protocol that enables agents to develop opinions about the trustworthiness of the other agents in the system in directed graphs. Next, we provide the necessary definitions from these works that will be referenced throughout.

\subsection{Trust Opinions}
Following previous works~\cite{gil2017guaranteeing,ourTRO,yemini2022resilience,L4DC,L4DC_2023_extended, gil2023physicality} that use and develop the stochastic observation of trust, we give the following definition:
\begin{definition}[Stochastic Observation of Trust $\aij$]
    \label{defn:alpha_value}
    If agent $j\in \mV$ sends information to a legitimate agent $i\in \mL$ at time $k$, meaning that we have $j \in \Nin$, agent $i$ receives a stochastic observation of trust $\aij[k]$ which is the likelihood that agent $j$ is trustworthy. We assume that $\aij[k] \in[0,1]$ for all $k$. 
\end{definition}

In our previous work \cite{L4DC, L4DC_2023_extended}, we developed a learning protocol where agents can develop opinions about the trustworthiness of all the agents in the system using the trust observations and the opinions of their trusted neighbors. Moreover, we showed that these opinions converge to the true trustworthiness of agents over time. In this work, we employ the same protocol.
\begin{definition}[Opinion of Trust]
    \label{defn:opinion}
    Let $o_{ij}[k]\in[0,1]$ denote agent $i$'s opinion of its trust about agent $j$ at time $k$. A legitimate agent $i \in \mL$ trusts agent $j$ at time $k$ if $o_{ij}[k]\geq1/2$ and does not trust agent $j$ otherwise. Moreover, we stack these opinions to define a vector $o_i[k]$ that stores agent $i$ opinions about all the other agents. 
\end{definition}
We define the aggregate trust value of agent $i\in \mL$ about agent $j\in \Nin$ at time $k$ as $\bij[k] \triangleq \sum_{t=0}^{k-1} (\aij[t]-0.5)$ for $k\geq 1$ and define $\bij[0]=0.$ Also, we define $\bij[k]=1$ for all $k$. During every communication round, agents share their opinion vectors with each other. A legitimate agent $i$ determines its opinion about an in-neighbor $j\in \Nin$ as $o_{ij}[k]=\mathbf{1}_{\{\bij[k]\geq 0\}}$, where $\mathbf{1}$ is the indicator function. Using this, we define the trusted in-neighborhood of agent $i$ at time $k$ as $\Nin[k]=\{j\in \Nin \mid o_{ij}(t)\geq 1/2\}.$ For an arbitrary agent $q$ that is not an in-neighbor, agent $i$ uses the opinions of its trusted in-neighbors to update its opinion as
$$
    o_{iq}[k]=\sum_{j \in \Nin[k]} \frac{o_{jq}[k-1]}{|\Nin[k]|}.
$$
Using these opinions, we also define the trusted out-neighborhood of an agent $i$ at time $k$ as $\Nout[k]=\{j\in \Nout \mid o_{ij}(t)\geq 1/2\}.$

We make the following assumptions about the trust observations.
\begin{assumption}[Trust Observations]\label{assumption:trust-observations}
Assume that
\begin{enumerate}
    \item[i)] {\rm [Difference of trust observations in expectation]}. The expectation of the
variables $\aij[k]$ are constant for malicious transmissions and legitimate transmissions, respectively, i.e., for some scalars $E_{\mM}$, $E_{\mL}$ with $E_{\mM}<0$ and $E_{\mL}>0$, $E_{\mM} = \mathbb{E}[\aij[k]]-1/2$ for all $i \in \mL, \ j \in \Nin \cap \mM$, and $E_{\mL} = \mathbb{E}[\aij[k]]-1/2$ for all $i \in \mL, \ j \in \Nin \cap \mL.$
    \item[ii)] {\rm [Independence of trust observations.]} The observations $\aij[k]$ are independent for all $k$. Moreover, for any $i\in \mL$ and $j \in \Nin$, the observation sequence $\{\aij[k]\}_{k\in \mathbb{N}}$ is identically distributed.
\end{enumerate}
\end{assumption}
We note that the homogeneity of the trust variables (\Cref{assumption:trust-observations}(i)) and identically distributed observation sequence $\{\aij[k]\}_{k\in \mathbb{N}}$ assumptions (\Cref{assumption:trust-observations}(ii)) are introduced for the sake of simplicity of the presentation. These assumptions can be relaxed to cover cases where the trust observations are heterogeneous over time and for different pairs of agents $(i,j)$.

Finally, we make the following assumptions on the connectivity of the communication network.
\begin{assumption}[Connectivity of Network]\label{assumption:graph-connectivity}
\
\begin{enumerate}
    \item[i)] {\rm [Sufficiently connected graph].} 
    The subgraph $\GL$ induced
by the legitimate agents is strongly connected.
    \item[ii)] {\rm [Observation of malicious agents].} 
    For any malicious agent $j \in \mM$, there exists some legitimate agent $i \in \mL$ that observes $j$, i.e., $j \in \Nin$ for some $i \in \mL$.
\end{enumerate}
\end{assumption}

The sufficient connectivity assumption (\Cref{assumption:graph-connectivity}(i)) is common in the literature of resilient distributed optimization. The observation of malicious agents assumption (\Cref{assumption:graph-connectivity}(ii)) is required for learning the trustworthy out-neighbors of the legitimate agents as noted in \cite{L4DC,L4DC_2023_extended}.

\subsection{Problem Definition}
Our goal in this work is to develop a distributed optimization algorithm to solve the problem given in \eqref{eq:objective_defn} in the presence of malicious agents. More specifically, we want to solve the following problems:
\begin{problem}\label{problem:as_convergence}
    Let $x_i[k]$ denote the estimate of agent $i\in \mL$ for the solution to the optimization problem given in \eqref{eq:objective_defn}. We aim to develop a distributed optimization algorithm such that the iterates $x_i[t]$ generated by the algorithm converge to the optimal point $x^*$ for all legitimate agents $i\in \mL$ almost surely.
\end{problem}

\begin{problem}\label{problem:expected_convergence}
    We want to characterize the convergence rate of the expected error $\norm{x_i[k]-x^*}^2$ for all legitimate agents $i\in \mL.$
\end{problem}

\section{Resilient Projected Push-Pull (RP3) Algorithm}
\subsection{Background}\label{sec:background}
As we are dealing with solving problem~\eqref{eq:objective_defn}
over directed graphs, 
we base our algorithm on the Projected Push-Pull algorithm \cite{projected_push_pull} developed for constrained distributed optimization problems over time-varying graphs.
In the Projected Push-Pull algorithm, agents store two decision variables $x_i[k]$ and $z_i[k]$ and a gradient tracking variable $y_i[k]$. Agents initialize $x_i[0]=z_i[0]\in \mX$ arbitrarily, and $y_i[0]=\nabla f_i(x_i[0])$. Agents share their variables $z_i[k]$ and the scaled gradient tracking variables $C_{ij}y_i[k]$ with their out-neighbors at every communication round $k$ and do the following updates:
\begin{subequations}\label{eq:PPP}
    \begin{align}
    x_{i}[k+1]&= \sum_{j=1}^n R_{ij}z_{j}[k],\label{eq:PPP_update_x}\\
    y_{i}[k+1]&= \sum_{j=1}^n C_{ij}y_j[k]+\nabla f_i(x_i[k+1])-\nabla f_i(x_i[k]),\label{eq:PPP_update_y}\\
    z_i[k+1]&=(1-\lambda)x_i[k+1]+\lambda \projx{x_i[k+1]-\eta y_i[k+1]},\label{eq:PPP_update_z}
\end{align}
\end{subequations}
where $\eta$ and $\lambda$ are the step-sizes for the algorithm. Agents choose the weights $R_{ij}$ such that $R_{ij}>0$ if and only if $j\in \Nin$ and $\sum_{j\in \Nin}R_{ij}=1.$ Similarly, $C_{ji}>0$ if and only if $j \in \Nout$ and $\sum_{j\in \Nout}C_{ji}=1.$ These choices of weights results in a row stochastic weight matrix $R$ with $ij$th element $R_{ij}$ and column stochastic matrix $C$ with $ij$th element $C_{ij}$. The Projected Push-Pull satisfies $\sum_{i=1}^n y_i[k] = \sum_{i=1}^n \nabla f_i(x_i[k]),$ at each time step $k,$ which is called the gradient tracking property. This property depends on the \textit{correct initialization} of $y_i$ variables and it is crucial for the convergence of the algorithm in all gradient tracking methods \cite{harnessing_smoothness_journal,push-pull,AB_usman}.

The Projected Push-Pull algorithm is especially suitable for our problem for the following reasons 1)~it is compatible with directed communication graphs, 2)~it achieves a geometric convergence rate, and 
3)~the $z_i[k]$ variables shared among the agents stay within the feasible region $\mX$, limiting the effect of malicious agents on the system. However, the algorithm is not designed for handling malicious agents and we will describe the necessary modifications in the next section. Still, for completeness we discuss the analysis and convergence of this algorithm without the malicious agents in this section.

Denote the left eigenvector of the row stochastic matrix $R$ corresponding to the eigenvalue $1$ by $\phi$, meaning that we have $\phi^\intercal R = \phi$. Similarly, denote the right eigenvector of the column stochastic matrix $C$ corresponding to the eigenvalue $1$ by $\pi$, meaning that we have $C \pi = \pi$. Here, both $\phi$ and $\pi$ are stochastic vectors with positive entries. In the analysis of this algorithm, there are three different error terms that we keep track of: 1)~optimality error, 2)~the consensus error, and 3)~the gradient tracking error. We define these respective  error terms mathematically as follows:
\begin{align}
    \intertext{\textbf{Optimality error:}}
     \normphi{\boldx[k]-\boldx^*}\triangleq\sqrt{\sum_{i=1}^n \phi_i \normsq{x_i[k]-x^*}},
     \label{eq:optimality-error}
\end{align}
where $\boldx[k]=(x_1[k], \dots, x_n[k]),$ $\boldx^*=(x^*, \dots, x^*).$
\begin{align}
    \intertext{\textbf{Consensus error:}}
    D(\boldx[k], \phi)\triangleq\sqrt{\sum_{j=1}^n\sum_{i=1}^n\phi_i \phi_j \norm{x_i[k]-x_j[k]}^2}.
    \label{eq:consensus-error} 
\end{align}
\begin{align}
    \intertext{\textbf{Gradient tracking error:}}
    S(\boldy[k], \pi)\triangleq \sqrt{\sum_{i=1}^n \pi_i \norm{\frac{y_i[k]}{\pi_i}-\sum_{l=1}^n y_l[k]}^2},
    \label{eq:gradient-tracking-error}
\end{align}
where $\boldy[k]=(y_1[k], \dots, y_n[k])$. 

Convergence of the Projected Push-Pull is characterized using several parameters that depend on the properties of the communication graph $\mG$, and the matrices $R$ and $C$. We define additional notation to introduce these parameters. Let $\min (v)$ and $\max (v)$ denote the minimum and maximum values of a vector $v$, respectively. Also, let $\min M^+$ denote the minimum non-zero value of a non-negative matrix $M$. We represent the diameter and the maximum edge utility of graph $\mG$ with $\mathsf D(\mG)$ and $\mathsf K(\mG)$, respectively (see \cite[Lemma 6.1]{nguyen2022distributed} for a detailed definition). We define
\begin{align*}
    \sigma &\triangleq \sqrt{1-\frac{\min(\phi) (\min R^+)^2}{\max^2 (\phi) \mathsf D(\mG) \mathsf K(\mG)}}\in (0,1),\\
    \tau &\triangleq \sqrt{1- \frac{\min^2 (\pi) (\min C^+)^2}{\max^3(\pi_) \mathsf{D}(\mG) \mathsf K(\mG)}}\in (0,1), \\
    r &\triangleq \sqrt{\frac{1}{\min (\pi)}} + \sqrt{n}, \text{ and }
    \varphi \triangleq \sqrt{\frac{1}{\min (\phi)}} + \sqrt{n}.
\end{align*}
Here, $\sigma$ and $\tau$ are the contraction coefficients we get from the matrices $R$ and $C$, respectively. Typical analyses of gradient tracking methods involve upper bounding the error terms at the $k+1$th time step in terms of the errors at the $k$th time step. The relationship between the error terms depends on the system parameters we defined, as well as the strong convexity and $L$-smoothness of the cost functions and the step sizes involved in the algorithm. By expressing these relationships in a system of inequalities, it can be demonstrated that the algorithm converges with carefully chosen step sizes. For more details on this type of analysis, see \cite{push-pull,projected_push_pull,robust2020}. We state the following result characterizing the convergence of the Projected Push-Pull algorithm.

\begin{theorem}[Theorem 1, \cite{projected_push_pull}]
    Define the error vector $\bolde[k]=(\normphi{\boldx[k] - \boldx^*}, D(\boldx[k], \phi), S(\boldy[k], \pi))^\intercal.$  Let \Cref{assumption:cost_function}, \Cref{assumption:closed_convex_const_set}, and \Cref{assumption:graph-connectivity} hold. Let $0<\eta<\frac{1}{nL}$ and 
    \begin{align*}
    \lambda < \min \left\{\frac{1-\sigma}{2\varphi \sqrt{n}}, \frac{1-\tau}{r\varphi}, \frac{\eta n \min(\pi) \mu (1-\sigma)(1-\tau)}{K}\right\},
\end{align*} where
\begin{align*}
    K=
    &(1+\eta n \min(\pi) \mu) \varphi \cdot  
    [2\sqrt{n}(1-\tau)+r(1-\sigma) +2r (1+\sigma)].
\end{align*}
    Then, we have
    \begin{align}
        \bolde[k+1] \leq M(\eta, \lambda) \bolde[k],
        \label{eq:composite-error-relation}
    \end{align}
    where the inequality is elementwise and  $M(\eta, \lambda) \in \mathbb{R}^{3\times 3}$ is equal to
    \begin{align}
    \begin{bmatrix}
            1-\eta \lambda n \min(\pi) \mu & \lambda \varphi \sqrt{n} & \lambda  L^{-1}  \\
            2\lambda & \sigma + 2\lambda \sqrt{n} \varphi & 2 \lambda  L^{-1} \\
            2 \lambda Lr \varphi & Lr\varphi (1+\sigma + \lambda \varphi \sqrt{n}) & \tau + \lambda r \varphi 
        \end{bmatrix}
        \label{eq:matrix_defn}
\end{align}
    The spectral radius of $M(\eta, \lambda)$ is less than $1$, i.e., $\rho(M(\eta, \lambda))<1,$ where $\rho(\cdot)$ denotes the spectral radius of a matrix. Moreover, the errors $\normphi{\boldx[k] - \boldx^*}, D(\boldx[k],\phi),$ and $S(\boldy[k], \pi)$ all converge to $0$ geometrically fast with rate $\rho(M(\eta, \lambda))$.
\label{theorem:nominal_convergence}
\end{theorem}

We note that since Theorem \ref{theorem:nominal_convergence} requires \Cref{assumption:closed_convex_const_set} (i.e., closed and convex constraint set) and not the more restrictive  \Cref{assumption:compact_convex_const_set} which further requires that the constraint set is bounded. Nevertheless, since compactness introduces a natural bound on the impact that malicious agents can have on the decision variables, it is easier to analyze and we will first consider the case \Cref{assumption:compact_convex_const_set} holds true. Then, in Section \ref{sec:unbounded_case} we will extend our results to the case where \Cref{assumption:closed_convex_const_set} holds true.

Note that the Projected Push-Pull algorithm only converges when there are no malicious agents in the system. This is because the analysis of the algorithm in \cite{projected_push_pull} is based on several key assumptions: all the agents adhere to the update rule in~\eqref{eq:PPP}, the mixing matrix $R$ is row stochastic, $C$ is column stochastic, and the gradient tracking variables $y$ are correctly initialized. However, these assumption are violated in the presence of malicious agents. Such agents can transmit arbitrary data to their neighbors, thereby violating the row and column stochasticity of the matrices. Moreover, even if legitimate agents eventually identify and exclude all malicious agents, they would need to restart the algorithm to re-establish the gradient tracking property. Furthermore, as we will elaborate later on, agents do not know when their trust estimations are accurate, making it impossible to determine an appropriate restart time to guarantee convergence. In the next section, we present the resilient version of this algorithm that resolves these issues. 

\subsection{Algorithm}
In this section, we present the Resilient Projected Push-Pull (RP3) Algorithm given in \cref{alg:RP3}. 
\begin{algorithm}
\caption{Resilient Projected Push-Pull (RP3)}
\begin{algorithmic}[1]
\REQUIRE Optimization parameters $\eta, \lambda$, chosen according to \Cref{theorem:nominal_convergence}.
\STATE Each legitimate agent $i$ does the following:
\STATE Initialize $x_i[0]=z_i[0] \in \mX$ arbitrarily, and set $s_i[0]= 0$. 
\WHILE{$k=0, 1, \dots$}
\STATE Update trust opinions $o_{ij}(t)$ using the learning protocol and stochastic trust observations as shown in \Cref{defn:opinion}.
\STATE Determine the trusted in and out neighborhoods using the rule $\Nin[k]=\{j\in \Nin \mid o_{ij}(t)\geq 1/2\}$ and $\Nout[k]=\{j\in \Nout \mid o_{ij}(t)\geq 1/2\}.$
\STATE Determine coefficients $C_{ji}[k]$ and $R_{ij}[k]$ for in and out neighbors based on $\Nin[k]$ and $\Nout[k]$.
\STATE Send $z_i[k]$, $C_{ji}[k]y_i[k]$ to out-neighbors $j \in \mN_{i}^{out}.$
\STATE Receive $z_j[k]$, $C_{ij}[k]y_j[k]$ from in-neighbors $j \in ~ \mN_{i}^{in}.$
\STATE Perform the gradient tracking update using~\eqref{eq:RP3_update_s}:
\begin{align*}
    s_i[k+1] \gets &\projsk{\sum_{j=1}^n C_{ij}[k]s_j[k] + \nabla f_i(x_i[k])}.
\end{align*}
\STATE Perform the consensus update using~\eqref{eq:RP3_update_x}:
$$x_i[k+1] \gets \sum_{j=1}^n R_{ij}[k]z_j[k].$$
\STATE Perform the lazy update using~\eqref{eq:RP3_update_z}:
\begin{align*}
    z_i[k+1] \gets &(1-\lambda) x_i[k+1] +  \lambda \projx{x_i[k+1] - \eta s_i[k+1]}.
\end{align*}
\ENDWHILE
\end{algorithmic}
\label{alg:RP3}
\end{algorithm}
 The legitimate agents keep track of two decision variables $x_i[k]$ and $z_i[k]$, and a gradient tracking variable $s_i[k]$. Agents initialize $x_i[0]=z_i[0]\in \mX$ arbitrarily and choose $s_i[0]=0.$ Agents share their variables $z_i[k]$ and $C_{ji}[k]s_i[k]$ with their trusted out-neighbors at every communication round $k$ and do the following updates, for a predefined sequence of sets $\mS_k, k=0,1,\ldots$ which we will strategically choose:
\begin{subequations}\label{eq:RP3}
    \begin{align}
    s_i[k+1] =& \projsk{\sum_{j=1}^n C_{ij}[k]s_j[k] + \nabla f_i(x_i[k])}, \label{eq:RP3_update_s}\\
    x_{i}[k+1]=& \sum_{j=1}^n R_{ij}[k]z_{j}[k], \label{eq:RP3_update_x}\\
    z_i[k+1]=&(1-\lambda)x_i[k+1] 
    +\lambda \projx{x_i[k+1]-\eta (s_i[k+1]-s_i[k])}. \label{eq:RP3_update_z}
    \end{align}
\end{subequations}
Here, $\eta>0$ and $\lambda \in (0, 1]$ are two different step sizes. Recall that we define the trusted in-neighborhood and trusted out-neighborhood as $\Nin[k]=\{j\in \Nin \mid o_{ij}(t)\geq 1/2\}$ and $\Nout[k]=\{j\in \Nout \mid o_{ij}(t)\geq 1/2\}$, respectively. Legitimate agents choose the weights $R_{ij}[k]$ such that $R_{ij}[k]>0$ if and only if $j\in \Nin[k]$ and $\sum_{j\in \Nin[k]}R_{ij}[k]=1.$ Similarly, $C_{ji}[k]>0$ if and only if $j \in \Nout[k]$ and $\sum_{j\in \Nout[k]}C_{ji}[k]=1.$ The sequence of sets $\mS_k$ is not defined yet but we will elaborate on the choice of this set sequence later on. A malicious agent $m\in \mM$ can send any $z_m[k]\in \mX$ at all time $k$. However, since sending a value outside the set would reveal their maliciousness immediately, we restrict the $z$ values malicious agents send to the set $\mX$. 
Similarly, a malicious agent $m$ can only send $s_m[k+1]$ values that lie within the set $\mS_k$ at time $k$.

The proposed RP3 algorithm has three important modifications over the Projected Push-Pull algorithm \eqref{eq:PPP}, in the remaining of this section we will discuss these algorithmic choices. 
\subsubsection{Trust-based weights}
Agents assign positive weights to their trusted neighbors only. Using our results from \cite{L4DC,L4DC_2023_extended}, we will show that these weights will eventually stabilize and agents will assign positive weights to their legitimate neighbors only. 
\subsubsection{Preserving the gradient tracking property in the presence of malicious agents}
As we discuss in \cref{sec:background}, the convergence of the Projected Push-Pull algorithm depends on the initialization of the gradient tracking variables $y_i$. If these variables are not initialized such that $y_i[0]\neq\nabla f_i(x_i[0])$, then the gradient tracking property no longer holds. Consider the scenario where the malicious agents are removed from the system eventually, but they can affect the system for some time. In this case, there is no guarantee that the gradient tracking property will hold with the update rule~\eqref{eq:PPP_update_y}.
Therefore, we replace the gradient tracking update rule~\eqref{eq:PPP_update_y} in Projected Push-Pull with~\eqref{eq:RP3_update_s}. This change is adapted from the Robust Push-Pull algorithm for unconstrained problems with noisy communication links \cite{robust-pp}. Let us assume that the legitimate agents have some arbitrary $s_i[k]\in \R^d$. In the next communication step, if there are no malicious agents in the system and the legitimate agents follow the update rule \eqref{eq:RP3_update_s}, we have $\sum_{i\in \mL} s_i[k+1] - \sum_{i\in \mL} s_i[k] = \sum_{i\in \mL} \nabla f_i(x[k]),$ which means that the system will restore the gradient tracking property.

\subsubsection{Projecting the gradient tracking variables onto the set $\mS_k$}
An important design consideration in our problem is limiting the effect of the malicious agents until the legitimate agents cut them off from the system. The effect of the malicious agents on the $z_i$ variables is limited since all $z_i[k]\in \mX$ for all $i\in \mV$ and for all $k$, when $X$ is compact.  However, that is not the case for the gradient tracking variables $s_i[k]$. Thus, we employ the projection onto the set $\mS_k$ at every time $k$. The main concern with this change is to be able to preserve the gradient tracking property. In~\Cref{sec:bounding_s_variables} we show that by choosing a growing sequence of $\{\mS_k\}$ with an appropriate growth rate, the gradient tracking property will be restored.

\section{Analysis}
In this section, we introduce the theoretical foundations and convergence results of the RP3 algorithm. Our analysis hinges on the selection of appropriate mixing weights and a set sequence $\{\mS_k\}$, which together ensure that the RP3 method eventually becomes equivalent to the Projected Push-Pull algorithm without malicious agents initialized at some random point. This set sequence is critical as it limits the influence of malicious agents until agents' opinions of trust become reliable. 
Our primary challenge in designing the algorithm is that the agents cannot definitively know when their trust opinions are accurate, so they cannot simply wait until their trust opinions become accurate to start the algorithm at that time. 
Nonetheless, we will prove the existence of such time and use this concept in our analysis. We begin by formally defining the nominal behaviour of the RP3 method when it becomes equivalent to the Projected Push-Pull algorithm.

\begin{definition}[The nominal behavior of the RP3] 
    Assume that there exist a time $k'$ such that for all $k\geq k'$, $R_{ij}[k]=\Bar{R}_{ij}$ where $\Bar{R}_{ij}>0$ only if $j\in \Nin\cap \mL$ and $\sum_{j\in \Nin\cap \mL}\Bar{R}_{ij}=1.$ Similarly, for all $k\ge k'$, we have $C_{ij}[k]=\Bar{C}_{ij}$ where $\Bar{C}_{ji}>0$ only if $j\in \Nout\cap \mL$ and $\sum_{j\in \Nout\cap \mL}\Bar{C}_{ji}=1.$ This corresponds to the ideal case where legitimate agents assign positive weights to their legitimate neighbors only and malicious agents are excluded from the system. Also, assume that for all $k\geq k'$, we have 
    \begin{align*}
    \projsk{\sum_{j=1}^n \Bar{C}_{ij}s_j[k] + \nabla f_i(x_i[k])}  
        =\sum_{j=1}^n \Bar{C}_{ij}s_j[k] + \nabla f_i(x_i[k]),
    \end{align*} i.e., the projector operator onto the set $S_k$ becomes the identity operator. We call the behavior of~\Cref{alg:RP3} after such time $k'$ the nominal behavior of the algorithm. In this case, \cref{alg:RP3} becomes equivalent to the Projected Push-Pull algorithm given in~\eqref{eq:PPP} if we define $y_i[k]=s_i[k]-s_i[k-1].$
    \label{def:nominal_behavior}
\end{definition}
Our first goal in the analysis is to show that \cref{alg:RP3} reaches nominal behavior and converges from there onwards. Then, we will show that the algorithm reaches this behavior quickly while the effects of the malicious agents until reaching this behavior is bounded.

\subsection{Preliminary Results: Learning the Trustworthiness of the Agents}

In this part, we present the following results from~\cite{L4DC_2023_extended, L4DC} that will be used in our analysis.

\begin{lemma}[\cite{L4DC}, Corollary 1] \label{lemma:learning_trustworthiness}
Let \Cref{assumption:trust-observations} and \Cref{assumption:graph-connectivity} hold. Then, all legitimate agents $i\in \mL$ can learn the trustworthiness of all agents in the network correctly almost surely. That is, there exists a finite random time $\tmax$ such that for all $k\geq \tmax$ and for all $q\in \mV$, $o_{iq}(t) \geq 1/2$ if $q\in \mL$ and $o_{iq}(t) < 1/2$ if $q\in \mM$ almost surely.
\end{lemma}
We note that the time $\tmax$ is stochastic, but finite almost surely. The following corollary is a consequence of the choice of the weight matrices and Lemma~\ref{lemma:learning_trustworthiness}.

\begin{corollary}\label{corrollary:weight_convergence}
Let~\Cref{assumption:trust-observations} and \Cref{assumption:graph-connectivity} hold. Then, we have $R[k]=\Bar{R}$ for all $k\geq \tmax$ with weights $\Bar{R}_{ij}$ such that $\Bar{R}_{ij}>0$ if and only if $j\in \Nin\cap \mL$ and $\sum_{j\in \Nin\cap\mL}\Bar{R}_{ij}=1.$ Similarly, for all $k\geq \tmax$, we have $C[k]=\Bar{C}$ with weights $\Bar{C}_{ji}>0$ if and only if $j \in \Nout\cap \mL$ and $\sum_{j\in \Nout \cap \mL}\Bar{C}_{ji}=1.$
\end{corollary}
\Cref{corrollary:weight_convergence} shows that after reaching the time $\tmax$, legitimate agents will assign weights to their legitimate neighbors only, which is necessary for reaching the nominal behavior. It is not known when the system will reach this time since $\tmax$ is stochastic; its probability is characterized in \cite{L4DC_2023_extended}. First, we let $N_{\mL}$ be the total number of legitimate in-neighbors in the system, i.e., $N_{\mL} \triangleq \sum_{i\in \mL} |\Nin \cap \mL|.$ Similarly, we let  $N_{\mM}$ be the total number of malicious in-neighbors in the system, i.e., $N_{\mM} \triangleq \sum_{i\in \mM} |\Nin \cap \mM|.$ The following proposition characterizes some probabilities related to $\tmax$.

\begin{proposition}[Proposition 1, \cite{L4DC_2023_extended}]
    Define
    \begin{align}
        p_c(k) \triangleq 
N_{\mL}\exp(-2kE_{\mL}^2)+N_{\mL}\exp(-2kE_{\mM}^2),
        \label{eq:pc}
    \end{align}
    where $E_{\mL} \triangleq \mathbb{E}[\aij(t)]-1/2$ for $i\in \mL$ and $j \in \mL$ and $E_{\mM} \triangleq \mathbb{E}[\aij(t)]-1/2$ for $i\in \mL$ and $j \in \mM$ as defined in \Cref{assumption:trust-observations}.
    Also, define 
    \begin{align}
        p_e(k) \triangleq N_{\mL}\frac{\exp(-2kE_{\mL}^2)}{1-\exp(-2E_{\mL}^2)}+N_{\mM}\frac{\exp(-2kE_{\mM}^2)}{1-\exp(-2E_{\mM}^2)}.
        \label{eq:p_e}
    \end{align}
    Let $d_{max}$ denote the maximum in-degree of any legitimate node in graph $\mG,$ i.e., $d_{max} \triangleq~\max_{i\in \mL} |\Nin|.$ Let $\mathsf D(\mG)$ denote the diameter of the graph $\mG.$ Define $\Delta \triangleq h \cdot \mathsf D(\mG)+1$, where $h = 1/\log_2 \frac{1}{1-(1/d_{max})^{\mathsf D(\mG)}}.$ Then, we have for all $k\geq 0,$
    \begin{align}
        \Pr(T_{max}=k) &\leq \min \{p_c(k-\Delta),1\}, \text{ and } \\
        \Pr(T_{max}>k-1) &\leq \min \{p_e(k-\Delta),1\}.
    \end{align}
    \label{proposition:tmax_prob} 
\end{proposition}
Proposition~\ref{proposition:tmax_prob} will be particularly useful when deriving the expected convergence rate of the algorithm. Next, we will examine how the choices of the set sequence $\mS_k$ affect the algorithm.

\subsection{Bounding the $s$-variables}\label{sec:bounding_s_variables}
In the preceding section, we showed that the system will reach a time $\tmax$ after which the malicious agents will be effectively excluded from the dynamics. Therefore, our next goal is to show that the effect of including malicious agents or excluding legitimate agents before reaching this time is limited. In this section, we focus on bounding the effect of the malicious agents with a strategic choice of the set sequence $\{\mS_k\}.$ However, projecting the $s_i[k]$ variables onto the set sequence $\{\mS_k\}$ breaks the gradient tracking property. Therefore, we should choose a set sequence $\{\mS_k\}$ that grows fast enough such that the projections onto this set do not change $s_i[k]$, i.e., the projection becomes the identity operator after some time so as  to restore the gradient tracking property. We start by estimating the growth of $s_i[k]$ variables when the system is in the nominal behaviour, i.e., when there is no projection onto the set $\mS_k$ and no malicious agents are in the system. Then, we show that the rate of the growth of of $s_i[k]$ is bounded.

\subsubsection{$s$-variables growth
}
We start by expressing the update rule~\eqref{eq:RP3_update_s} in the nominal case
\begin{align}\label{eq:s_i_dynamic_k}
    s_i[k+1] &= \sum_{j=1}^n \Bar{C}_{ij} s_j[k] + \nabla f_i(x_i[k]).
\end{align}
Generally, the $s_i[k]$-values might not be bounded even when there are no malicious agents. In fact, they are not bounded if $\nabla f(x^*) \neq 0$. To see why, let us evaluate the following identity which follows from \eqref{eq:s_i_dynamic_k} and the column stochasticity of $\Bar{C}$:
\begin{align*}
    \boldone^T \bolds[k+1]= \boldone^T \bolds[k] + \boldone^T \nabla F(\boldx[k]),
\end{align*}
where $\bolds[k]=(s_1[k], \dots, s_n[k])$ and $\nabla F(\boldx[k])=(\nabla f_1(x_1[k]), \dots, \nabla f_n(x_n[k]))^T$.
If all agents converge to $x^*$, then we have
\begin{align*}
    \boldone^T \bolds[k+1] = \boldone^T \bolds[k] +  \boldone^T \nabla F(\boldx^*),
\end{align*}
{whose norm may grow since $\boldone^T \nabla F(\boldx^*)$ is added to the sum}. Next, we show that the growth is bounded.

\subsubsection{Finding a growing linear bound on $s$-values}
We show that there is a (linearly-growing) upper bound on $s_i[k]$ in the nominal case when the constraint set $\mX$ is bounded. This bound will be optimal as we have shown that the $s_i[k]$ variables grow at a linear rate when agents converge to $x^*$ and $\nabla f(x^*) \neq 0$.

The key to obtaining a bound for $s_i[k+1]$ is using the fact that the gradients of $f_i$ are bounded. We first state the following corollary of the $L$-smoothness of $f$ and compactness of $\mX$.
\begin{corollary}[Boundedness of Gradients]
    Let \Cref{assumption:cost_function} and \Cref{assumption:compact_convex_const_set} hold. Then, there exists $G\geq 0$ such that $\norm{\nabla f_i(x)}\leq G$ for all $x\in \mX$ and $i\in \mL$.
    \label{cor:gradient_bound}
\end{corollary}
\begin{proof}
Because $\nabla f_i$ 
are continuous and $\mX$ is compact, we can apply the extreme value theorem to upper bound $\norm{\nabla f_i(x)}$ over $\mX$. Then, we choose $G$ as the largest bound among all agents $i\in \mL.$
\end{proof}

The following result establishes the nominal growth rate of the $s_i[k]$ variables. 
\begin{proposition}
    \label{proposition:s_bound}
    Assume the RP3 algorithm has the nominal behavior since the beginning, i.e., legitimate agents assign non-zero weights to their legitimate neighbors only and there is no projection of the $s_i[k]$ variables. Let $s_i[0]=0$ for all $i \in \mL$. Then, $\norm{s_i[k]} \leq k n_\mL G$ for all $k \geq 0$ and $i\in \mL$.
    \label{proposition:s_growth_compact}
\end{proposition}

\begin{proof}
    We will first establish a recursion and then use the induction. For a legitimate agent $i\in\mL$, we have
    \begin{align*}
        \norm{s_i[k+1]}&= \norm{\sum_{j=1}^{n_\mL} \Bar{C}_{ij} s_j[k] + \nabla f_i(x_i[k])} \\
        &\leq \sum_{j=1}^{n_\mL} \Bar{C}_{ij}\norm{s_j[k]} + \norm{\nabla f_i(x_i[k])}\\
        &\leq \sum_{j=1}^{n_\mL} \Bar{C}_{ij}\norm{s_j[k]} + G.
    \end{align*}
    Then, by summing over all legitimate agents, we get
    \begin{align*}
        \sum_{i=1}^{n_\mL} \norm{s_i[k+1]} &\leq n_\mL G +\sum_{i=1}^{n_\mL} \sum_{j=1}^{n_\mL} \Bar{C}_{ij} \norm{s_j[k]} \\
        &=n_\mL G +\sum_{j=1}^{n_\mL} \left(\sum_{i=1}^{n_\mL} \Bar{C}_{ij}\right) \norm{s_j[k]} \\
        &\overset{(a)}{=} n_\mL G + \sum_{j=1}^{n_\mL} \norm{s_j[k]},
    \end{align*}
    where (a) follows from the fact that $\Bar{C}$ is a column stochastic matrix. Since $\sum_{i=1}^{n_\mL} \norm{s_i[0]} = 0$ due to the initialization, by induction and the definition of $s_i[k]$ in~\eqref{eq:s_i_dynamic_k} for
    the nominal case, it follows that $\sum_{i=1}^{n_\mL} \norm{s_i[k]} \leq k n_\mL G$. Thus, we have 
    \begin{align*}
        \max_{i\in\mL} \norm{s_i[k]} \leq \sum_{i=1}^{n_\mL} \norm{s_i[k]} \leq k n_\mL G,
    \end{align*}
    implying that all the norms $\|s_i[k]\|$ are bounded by $k n_\mL G$ in the nominal case.
\end{proof}

This result says that in the absence of malicious influence, the $s_i[k]$ variables would grow at most linearly with $k$. One might assume that legitimate agents can detect if an agent is malicious if it sends values greater than $k n_\mL G$. However, this growth rate is derived under the assumption that all the agents are legitimate at all time steps. In the presence of malicious agents, we lose the recursion we used in the proof of \cref{proposition:s_bound}. Therefore, we follow a different approach. Since the legitimate agents' $s_i[k]$ values cannot grow faster than linear in the nominal case, we can project the $s_i[k]$ values onto a set that expands at a rate surpassing the nominal growth. After some point, the projection will lose its effect and the algorithm will behave normally, restoring the gradient tracking property. Next, we formalize this intuition.

\subsubsection{Choosing the set sequence $\{\mS_k\}$} In \Cref{proposition:s_bound}, we considered scenarios where both legitimate and malicious agents are always classified accurately. We now return to the more general case of the RP3 algorithm, where the presence of malicious agents influences the system, and agents adjust their mixing weights based on trust opinions. The following result shows that the $s_i[k]$ values will stay in an invariant set with the correct choice of $\{\mS_k\},$ even when we have malicious agents in the system.
\begin{proposition}
\label{proposition:removing_projection}
Let $\theta > n_\mL G$ and $\mS_k = \{s \in \R^d \mid \norm{s} \leq \theta k\}$ for all $k\ge0$. For all legitimate agents $i\in \mL$ and all $k\ge0$, let 
\begin{align*}
    d_i[k+1]&=\sum_{j=1}^n C_{ij}[k] s_j[k] + \nabla f_i(x_i[k]),\\
    s_i[k+1] &=  \Pi_{\mS_k}\left[d_i[k+1]\right].
\end{align*}
Define $T_{nom} \triangleq \tmax n_\mL \frac{\theta - G}{\theta - n_\mL G}.$ Then, for all $k > T_{nom}$ and for all $i \in \mL$, we have $s_i[k]=d_i[k]$.
\end{proposition}
\begin{proof} 
Notice that we have $\norm{s_i[k]} \leq \norm{d_i[k]}$ for all $k \geq 0$. Then, for $k>\tmax$, we can write
    \begin{align*}
        \norm{d_i[k]} \leq \sum_{i\in \mL} \norm{d_i[k]} &\leq \sum_{i\in \mL} \norm{s_i[k]} \\
        &\leq n_{\mL}G + \sum_{i\in \mL} \norm{s_i[k-1]}  \\
        &\leq n_{\mL}G + \sum_{i\in \mL} \norm{d_i[k-1]} \\
        &\leq \dots \leq n_{\mL}G(k-\tmax) + \sum_{i\in \mL} \norm{s_i[\tmax]}\\
        &\leq n_{\mL} G(k-\tmax) + n_{\mL} \theta \tmax,
    \end{align*}
    where in the last step we used the fact that $\norm{s_i[\tmax]}\leq \theta \tmax$ due to the projection. When we have $k > n_{\mL} \frac{\theta - G}{\theta - n_{\mL} G} \tmax ,$ we get $n_{\mL} G(k-\tmax) + n_{\mL} \theta \tmax < \theta k.$ Therefore, 
    for all $k > T_{nom} = \tmax n_{\mL} \frac{\theta - G}{\theta - n_{\mL} G},$ we have that \begin{align*}
    \norm{d_i[k]} \leq n \theta \tmax + n_{\mL} G(k-\tmax)< \theta k.
\end{align*}
Hence, $d_i[k] \in \mS_k$ and $s_i[k]=\Pi_{\mS_k}\left[d_i[k]\right]=d_i[k]$.
\end{proof}
\begin{remark}
Proposition~\ref{proposition:removing_projection} states that the projection operator will be the identity operator after $k > T_{nom} \triangleq \tmax  n_{\mL} \frac{\theta - G}{\theta - n_{\mL} G}$. Hence, after this point, the protocol will have the gradient tracking property.
\end{remark}
\begin{corollary}
    Let $p_c(t)$ and $p_e(t)$ as defined in \eqref{eq:pc} and \eqref{eq:p_e}, respectively. Then, we have
    \begin{align}
        &\Pr(T_{nom}=t) \leq \min \left\{p_c\left(\frac{\theta - n_{\mL} G}{n_{\mL} (\theta - G)}t-\Delta\right),1\right\},\\
        &\text{ and} \nonumber \\
        &\Pr(T_{nom}\!>\!t-1) \leq \min \left\{p_e\!\left(\frac{\theta - n_{\mL} G}{n_{\mL}(\theta - G)}t-\Delta\right)\!,1\right\}.
    \end{align}
    \label{cor:Tnom_prob}
\end{corollary}
\begin{proof}
    The result follows directly from $T_{nom}=\tmax n_{\mL} \frac{\theta - G}{\theta - n_{\mL} G}$ and~\Cref{proposition:tmax_prob}.
\end{proof}
Notice that the legitimate agents do not need to know the value of $\tmax$ for these results to hold. However, we implicitly assumed that they know $n_{\mL}$ and $G$ or an upper bound on them while choosing $\theta$. However, as seen from the proof of \Cref{proposition:removing_projection}, we only need to ensure a faster growth than the nominal growth of the $s$-variables. Therefore, this assumption can be removed by choosing a set sequence that grows faster than linear. 
\begin{remark}
     Let $g(k)$ denote the maximum norms of the vectors in $\mS_k$, i.e., $g(k) \triangleq \norm{\mS_k}$. In \Cref{proposition:removing_projection}, $g(k)$ corresponds to $\theta k$. Agents can choose $g(k)$ that grows faster than linear, for example, $g(k)=k^2$. This way, agents do not need to know $n_{\mL}$ and $G$, and \cref{proposition:removing_projection} will hold with a different $T_{nom}.$
\end{remark}
For the clarity of the presentation, we will adhere to the choice of $\mS_k = \{s \in \R^d \mid \norm{s} \leq \theta k\}$ in our analysis. We will show results for exponentially growing sets in \Cref{sec:unbounded_case} and will discuss the impact of these different growth rates on the convergence of the algorithm.
\begin{remark}
    The norm we chose to define $\mS_k$ was Euclidean. Other norms can be used if they are more suitable for computations. However, all agents need to agree on the norm they use.
\end{remark}
In the next section, we provide our main results.
 
\subsection{Asymptotic Results}\label{sec:asymptotic_results_bounded}
Here, we present two of our main theorems, addressing \cref{problem:as_convergence}. First, we establish the almost sure convergence of Algorithm \ref{alg:RP3} to the optimal point. We then show the convergence of Algorithm \ref{alg:RP3} to the optimal point in the $r$th mean. To prove this convergence we present an auxiliary result which bounds the worst case error.
\begin{theorem}[Almost Sure Convergence]
Let $\mS_k = \{s \in \R^d \mid \norm{s} \leq \theta k\}$ with $\theta > n_{\mL} G$, and let each legitimate agent $i\in \mL$ initialize  $x_i[0],z_i[0]\in \mX$ arbitrarily and set $s_i[0]=0.$ 
Choose the stepsizes $\eta$ and $\lambda$ such that they satisfy the conditions defined in \Cref{theorem:nominal_convergence}. For each legitimate agent $i\in \mL$, denote the sequence generated by the dynamic \eqref{eq:RP3} by $\{x_i[k] \}$. Define the error vector \[\bolde[k]=(\normphi{\boldx[k] - \boldx^*}, D(\boldx[k], \phi), S(\boldy[k], \pi))^\intercal,\] and the random time
    $T_{nom}\triangleq \tmax n_{\mL} \frac{\theta - G}{\theta - n_{\mL} G}.$  
     Let Assumptions \ref{assumption:cost_function}, \ref{assumption:compact_convex_const_set}, \ref{assumption:trust-observations}, and \ref{assumption:graph-connectivity} hold true.
    Then, we have 
    \begin{align}
        \bolde[k] \leq M(\eta, \lambda)^{k-T_{nom}} \bolde[T_{nom}],
        \label{eq:err_bound_tnom}
    \end{align}
    for all $k>T_{nom}$ almost surely. Moreover, the sequence $\{x_i[k] \}$ converges to the optimal point $x^*$ for all $i\in \mL$ almost surely.
    \label{theorem:as_convergence}
\end{theorem}
\begin{proof}
    We start by showing that the algorithm reaches the nominal behavior at some finite time almost surely. By \cref{lemma:learning_trustworthiness}
    the finite (random) time $\tmax$ exists. Moreover, the weights almost surely converge to the correct weights defined in the nominal behavior by \cref{corrollary:weight_convergence}. 
    
    Define
    $T_{nom}\triangleq \tmax n_{\mL} \frac{\theta - G}{\theta - n_{\mL} G}.$ The agents will stop projecting their $s_i[k]$ values after reaching $T_{nom}$ as shown in \cref{proposition:removing_projection}. Denote the variables of the Projected Push-Pull algorithm given in \eqref{eq:PPP} with $x_i'[k]$, $z_i'[k]$, and $y_i'[k]$. Then, after time $T_{nom}$, \cref{alg:RP3} will be equivalent to running the Projected Push-Pull algorithm given in \eqref{eq:PPP} with the initialization $x_i'[0]=x_i[T_{nom}]$, $z_i'[0]=z_i[T_{nom}]$, and $y_i'[0]=s_i[T_{nom}]-s_i[T_{nom}-1]$. Therefore, by \cref{theorem:nominal_convergence}, we have
    $$\bolde[k] \leq M(\eta, \lambda)^{k-T_{nom}} \bolde[T_{nom}].$$
    Moreover, the sequence $\{x_i[k] \}$ generated by this dynamics converges to $x^*$ for all initial points $x_i'[0]\in \mX$ and $y_i'[0]=\nabla f_i(x_i'[0])$, which concludes the proof.
\end{proof}

Next, we will show the convergence of the algorithm to the optimal point in the $r$-th mean. We will use the Dominated Convergence Theorem \cite{nlar2011} in our proof. Before doing that, we will first bound each error term at the time the system reaches the nominal behavior. Since the influence of the malicious agents is still in the system before reaching the nominal behavior, these bounds reflect ``the worst case" scenario for the error terms. The following lemma provides these bounds.
\begin{lemma}
[The Worst Case Error Bounds]
\label{lemma:worst_case_err_bounds}
    Let $\mS_k = \{s \in \R^d \mid \norm{s} \leq \theta k\}$ and $B$ denote the bound on the vectors in $\mX$ as defined in \Cref{assumption:compact_convex_const_set}. Then, we have
    \begin{subequations}
    \begin{align}
        \normphi{\boldx[k] - \boldx^*} &\leq 2B,         \label{eq:error_bounds_x_compact}\\
        D(\boldx[k], \phi) &\leq 2B, \label{eq:error_bounds_D_compact}\\
        S(\boldy[k], \pi) &\leq \frac{2(n_{\mL}+1)}{\min (\pi)} \theta k, \label{eq:error_bounds_S_compact}
    \end{align}
    \end{subequations}
where $\min (\pi)$ denotes the minimum element of the stochastic vector $\pi$.
\end{lemma}
\begin{proof}
   Using the compactness of set $\mX$ and the triangular inequality, we obtain $\norm{x_i[k]-x^*} \leq 2B.$ Then, using the definition of $\normphi{\boldx[k] - \boldx^*}$ we obtain
   \begin{align*}
       \normphi{\boldx[k]-\boldx^*} &= \sqrt{\sum_{i\in \mL} \phi_i \normsq{x_i[k]-x^*}} 
       \leq  2B \sqrt{\sum_{i\in \mL} \phi_i} = 2B,
   \end{align*}
   where in the last step, we used the stochasticity of the vector $\phi_i$. The bound on $D(\boldx[k], \phi)$ is obtained similarly. Next, we bound the gradient tracking error $S(\boldy[k], \pi).$
   By definition, we have $$S(\boldy[k], \pi) = \sqrt{\sum_{i\in \mL} \pi_i \norm{\frac{y_i[k]}{\pi_i}-\sum_{l\in \mL} y_l[k]}^2}.$$ First, notice that $$\norm{y_i[k]}=\norm{s_i[k]-s_i[k-1]}\leq 2\theta k,$$ for any $i\in\mL$ due to the projection onto the set $\mS_k $. Using the triangular inequality, we get
    \begin{align*}
        \norm{\frac{y_i[k]}{\pi_i}-\sum_{l\in \mL} y_l[k]}^2
        &\leq \left(\norm{\frac{y_i[k]}{\pi_i}}+\norm{\sum_{l\in \mL} y_l[k]}\right)^2 \\
        &\leq \left(\frac{2(n_{\mL}+1)}{\min (\pi)} \theta k \right)^2.
    \end{align*}
    Using this bound and the stochasticity of $\pi$, we obtain $S(\boldy[k], \pi)\leq \frac{2(n_{\mL}+1)}{\min (\pi)} \theta k.$
\end{proof}

Now, we give the convergence in mean result which relies on the bounds \eqref{eq:error_bounds_x_compact}-\eqref{eq:error_bounds_D_compact}. We note that the additional upper bound \eqref{eq:error_bounds_S_compact} will be utilized in the finite time analysis of the convergence rate of the RP3 algorithm which we present later on in \Cref{theorem:exp_conv_rate}.
\begin{theorem}[Convergence in Mean]
    Let $\mS_k = \{s \in \R^d \mid \norm{s} \leq \theta k\}$  with $\theta > n_{\mL} G.$ Let Assumptions \ref{assumption:cost_function}, \ref{assumption:compact_convex_const_set}, \ref{assumption:trust-observations}, and \ref{assumption:graph-connectivity} hold true. Let each legitimate agent $i\in \mL$ initialize the algorithm such that $x_i[0],z_i[0]\in \mX$ arbitrarily and $s_i[0]=0.$ Choose the stepsizes $\eta$ and $\lambda$ such that they satisfy the conditions defined in \Cref{theorem:nominal_convergence}. Then, the sequence generated by the dynamic \eqref{eq:RP3} converges in the $r$-th mean to $\boldx^*$ for any $r\geq 1$, that is
    $$
        \lim_{k\to \infty} \E[\normphi{\boldx[k] - \boldx^*}^r]=0.
    $$
\label{theorem:convergence_in_mean}
\end{theorem}
\begin{proof}
    We will prove this theorem using the Dominated Convergence Theorem \cite{nlar2011}. Using the bounds given in \Cref{lemma:worst_case_err_bounds} we get
    \begin{align*}
        \normphi{\boldx[k] - \boldx^*} &\leq 2B, \\
        \normphi{\boldx[k] - \boldx^*}^r &\leq (2B)^r,
    \end{align*}
    where we take the $r$-th power of both sides. The error $\normphi{\boldx[k] - \boldx^*}^r$ is bounded by a constant value. Recall that we have almost surely convergence by \Cref{theorem:as_convergence}. Therefore, the desired result follows directly from the Dominated Convergence Theorem.
\end{proof}

\subsection{Finite Time Analysis}\label{sec:finite_time_results_bounded}
In this part, we derive the expected convergence rate of the algorithm.
\begin{theorem}[Expected Convergence Rate]
    Let $\mS_k = \{s \in \R^d \mid \norm{s} \leq \theta k\}$ with $\theta > n_{\mL} G.$ Let Assumptions \ref{assumption:cost_function}, \ref{assumption:compact_convex_const_set}, \ref{assumption:trust-observations}, and \ref{assumption:graph-connectivity} hold. Let each legitimate agent $i\in \mL$ initialize the algorithm such that $x_i[0],z_i[0]\in \mX$ arbitrarily and $s_i[0]=0.$ Choose the stepsizes $\eta$ and $\lambda$ such that they satisfy the conditions in \Cref{theorem:nominal_convergence}.
    Define the error vector $\bolde[k]=(\normphi{\boldx[k] - \boldx^*}, D(\boldx[k], \phi), S(\boldy[k], \pi))^\intercal$.  Then, for all $k\geq 0$, we have
    \begin{equation}
        \begin{aligned}
        \E[\bolde[k]] &\leq 
         M(\eta, \lambda)^{k-\lfloor k/2 \rfloor} (I-M(\eta, \lambda))^{-1} \begin{bmatrix}
            2B \\ 2B \\ \frac{2(n_{\mL}+1)}{\min (\pi)} \theta \lfloor k/2 \rfloor
        \end{bmatrix} \\
        &+\min \left\{p_e\left(\frac{\theta - n_{\mL} G}{n_{\mL} (\theta - G)}(\lfloor k/2 \rfloor+1)-\Delta\right),1\right\} 
        \cdot \begin{bmatrix}
            2B \\ 2B \\ \frac{2(n_{\mL}+1)}{\min (\pi)} \theta k 
        \end{bmatrix},
    \end{aligned}
    \label{eq:exp_conv_rate}
    \end{equation}
    where $B$ denote the bound on the vectors in $\mX$ as defined in \Cref{assumption:compact_convex_const_set}, $\Delta$ and $p_e(\cdot)$ are  as given in \Cref{proposition:tmax_prob}, and $\lfloor \cdot \rfloor$ denotes the floor function.
    \label{theorem:exp_conv_rate}
\end{theorem}
\begin{proof}
    We know that after time $T_{nom}$, the system will reach nominal behavior as shown in the proof of \cref{theorem:as_convergence}. Since we can analyze the system after reaching the nominal behavior, our strategy in the proof is to use the law of iterated expectations by conditioning on $T_{nom}.$  The main idea of the proof is twofold: For small realizations of $T_{nom}$, error reduction occurs as a result of contraction in the nominal case, facilitated by the early achievement of nominal behavior. Conversely, with a large $T_{nom}$, the error terms may increase according to their upper limits. Yet, the exponential decrease in the probability of $T_{nom}=k$ with increasing $k$ allows us to bound the expected error in this case. Hence, by the law of total expectation we have 
    \begin{align}
        \E[\bolde[k]]  &= \E[\E[\bolde[k]|T_{nom}]] \nonumber\\
                    &= \sum_{t=0}^{\lfloor k/2 \rfloor} \Pr(T_{nom}=t)\E[\bolde[k]|T_{nom}=t] \nonumber\\
                    &+ \Pr(T_{nom}> \lfloor k/2 \rfloor) \E[\bolde[k]|T_{nom}> \lfloor k/2 \rfloor].
                    \label{eq:law_iterative_exp}
    \end{align}
    We bound the first term as follows. Note that inequalities we use with respect to vectors and matrices hold entry-wise.
    \begin{align*}
        \sum_{t=0}^{\lfloor k/2 \rfloor} &\Pr(T_{nom}=t)\E[\bolde[k]|T_{nom}=k] \\
        &\overset{(a)}{\leq} \sum_{t=0}^{\lfloor k/2 \rfloor}  M(\eta, \lambda)^{k-t} \bolde[t] \\
        &\overset{(b)}{\leq} \left( \sum_{t=0}^{\lfloor k/2 \rfloor}  M(\eta, \lambda)^{k-t} \right) \begin{bmatrix}
            2B \\ 2B \\ \frac{2(n_{\mL}+1)}{\min (\pi)} \theta \lfloor k/2 \rfloor
        \end{bmatrix}. \\
    \end{align*}
    In inequality $(a)$, we bounded $\Pr(T_{nom}=t)$ with $1$, and used \eqref{eq:err_bound_tnom} given in \Cref{theorem:as_convergence}. Inequality $(b)$ follows directly from \Cref{lemma:worst_case_err_bounds} and the fact that $M(\eta, \lambda)$ is a non-negative matrix. Next, we bound the matrix summation in the last inequality.
    \begin{align}
        \sum_{t=0}^{\lfloor k/2 \rfloor}  M(\eta, \lambda)^{k-t} \nonumber
        &= M(\eta, \lambda)^{k-\lfloor k/2 \rfloor}\left( \sum_{t=0}^{\lfloor k/2 \rfloor}  M(\eta, \lambda)^{t} \right) \nonumber\\
        &\overset{(a)}{\leq} M(\eta, \lambda)^{k-\lfloor k/2 \rfloor} \left( \sum_{t=0}^{\infty}  M(\eta, \lambda)^{t} \right) \nonumber\\
        &\overset{(b)}{=} M(\eta, \lambda)^{k-\lfloor k/2 \rfloor} (I-M(\eta, \lambda))^{-1},
    \end{align}
    where in $(a)$, we used the non-negativity of the matrix $M(\eta, \lambda)$. Inequality $(b)$ comes from the infinite sum of matrices with spectral radius less than $1$ (see \cite[Theorem 3.15]{Varga2000}). 
    Now, we will bound the second term in \eqref{eq:law_iterative_exp}. Let $\Delta$ and $p_e(\cdot)$ be as given in \Cref{proposition:tmax_prob}. We have
    \begin{align*}
        &\Pr(T_{nom} > \lfloor k/2 \rfloor) \E[\bolde[k]|T_{nom}> \lfloor k/2 \rfloor] \\
        &\leq \Pr(T_{nom}> \lfloor k/2 \rfloor) \begin{bmatrix}
            2B \\ 2B \\ \frac{2(n_{\mL}+1)}{\min (\pi)} \theta k 
        \end{bmatrix} \\
        &\leq \min \left\{p_e\left(\frac{\theta - n_{\mL} G}{n_{\mL}(\theta - G)}(\lfloor k/2 \rfloor+1)-\Delta\right),1\right\} \\
        &\cdot \begin{bmatrix}
            2B \\ 2B \\ \frac{2(n_{\mL}+1)}{\min (\pi)} \theta k 
        \end{bmatrix},
    \end{align*}
    where the first inequality follows from \Cref{lemma:worst_case_err_bounds} and the second one follows from \Cref{cor:Tnom_prob}. Combining all the bounds gives us the desired result.
\end{proof}

\Cref{theorem:exp_conv_rate} states that for a sufficiently large $k$, the expected convergence rate of the system decays geometrically. The convergence rate depends on various properties of the system and design choices. First, the error contractions we get from the matrix $M(\eta, \lambda)$ depend on the choices of step sizes $\eta$ and $\lambda$, as well as the contractions we get from the matrices $\Bar{R}$ and $\Bar{C}$, and the smoothness and the convexity of the cost functions (see \eqref{eq:matrix_defn}). Second, both of the error terms depend on $B$ and a linear bound that grows over time. These terms reflect the effect that the malicious agents inflict before the system reaches the nominal behavior. Lastly, the second error term depends on the learning rate that we get from the learning protocol with trust opinions. The $\Delta$ term captures the impact of the graph topology on the learning rate while the coefficient $\frac{\theta - n_{\mL}G}{n_{\mL}(\theta - G)}$ captures the impact of the growth rate of the set $S_k$ on the time before reaching the nominal behavior. An interesting trade-off is that while a faster growth rate gives us a better decrease in the probability $p_e(\cdot)$, it also increases the impact that the malicious agents can have on the system through the gradient tracking terms $s_i[k].$ In the next section, we will see this impact for a different choice of $S_k.$

\section{Optimization Problems with Unbounded $\mX$}\label{sec:unbounded_case}


In this section, we extend our results to optimization problems with unbounded constraint sets. This setup is more challenging to guard against malicious behavior and to analyze, since the inputs of the malicious agents do not reside within a known compact set. To capture this, throughout this section, we will use the more general \Cref{assumption:closed_convex_const_set} instead of \Cref{assumption:compact_convex_const_set}. 

\subsection{Bounding the $x$- and $s$-variables}

In this setup the input values of the legitimate and malicious agents are not necessarily bounded and can take any choice in $\mathbb{R}^d$. The main challenge in removing the limitation on malicious agents to choose input values from a predefined bounded set is that these agents can arbitrarily influence variables $z_i[k]$, hence, the decision variables $x_i[k].$ Therefore, we need to confine their impact until the system reaches the nominal behavior. Our strategy, in this case, is to reapply the principle of projecting the gradient tracking variables $s_i[k]$ of growing bounded set to the agent's data values $x_i[k]$ as well. Specifically, for the $x_i[k]$ variables, we will introduce an expanding bounded set $\mX_k$ to bound the effect of the malicious agents. First, we show that introducing this set does not affect the convergence of the algorithm in the nominal case without the malicious agents. 
\begin{lemma}[Nominal Convergence with $\mX_k$]
    Define $\mX_k = \{x \in \R^d \mid \norm{x} \leq \exp (\theta k)\}$ with $\theta >0$ for all $k\geq 0$. Define the effective constraint set at time $k$ as $\Bar{\mX}_k\triangleq \mX \cap \mX_k.$ 
     Assume that there are no malicious agents in the system, and agents run the Projected Push-Pull algorithm given in \eqref{eq:PPP} with the the projections on the set $\Bar{\mX}_k$ instead of $X$ at every time step $k$. Then, \Cref{theorem:nominal_convergence} holds true after time $k'$, where $k'\triangleq \frac{1}{\theta} \ln \norm{x^*},$ where 
     $x^*$ is the optimal solution of the problem given in \eqref{eq:objective_defn}.
    \label{lemma:nominal_conv_unbounded}
\end{lemma}
\begin{proof}
    When $k'\triangleq {\frac{1}{\theta}} \ln \norm{x^*},$ then $x^* \in \mX_k$ for all $k\geq k'$ by the definition of $\mX_k$. Hence, we also have $x^* \in \Bar{\mX}_k.$ 
    Since $\Bar{\mX}_k \subseteq \mX$ and $x^* \in \Bar{\mX}_k $, the point $x^*$ is also the optimal solution of the problem in \eqref{eq:objective_defn} with the constraint set $\Bar{\mX}_k.$ The one-step contraction in the error given by \eqref{eq:composite-error-relation} only requires the constrained set to be closed and convex, i.e.,  satisfy \Cref{assumption:closed_convex_const_set}. Since $\Bar{\mX}_k$ is closed and convex and $x^*$ is included within the set after time $k'$, \Cref{theorem:nominal_convergence} holds. 
\end{proof}

Lemma~\ref{lemma:nominal_conv_unbounded} shows that these new growing sets only have a minimal impact on the convergence of the Projected Push-Pull algorithm, by introducing a short delay that has a logarithmic dependence on the norm of $x^*.$ Notice that even when $x^* \notin \mX_k$, the RP3 algorithm preserves the gradient tracking property when legitimate agents assign positive weights to their legitimate neighbors only. Therefore, in our forthcoming analysis where we derive a bound on $s_i[k]$ variable, we do not need to consider this time delay.

Our derivations of the growth rate of the variables $s_i[k]$ in \Cref{sec:bounding_s_variables} rely on the compactness of $\mX$ (see \Cref{proposition:s_bound}). 
Next, we will prove that this result also holds for the increasing sequence of sets $\Bar{\mX}_k$. This will later help us to find a way to construct the new growing set $\mS_k$ when the bound on $x_i[k]$ also grows over time. 

\begin{lemma}
    For some $\theta >0$ and for all $k\geq 0$, let
    $\mX_k = \{x \in \R^d \mid \norm{x} \leq \exp (\theta k)\}$. Also, let 
    $\Bar{\mX}_k = \mX \cap \mX_k$ and define $G_k \triangleq A \exp(\theta k)$ for all $k\ge0$, where $A \triangleq L+ \max_{i\in\cal{L}} \{\norm{\nabla f_i(0)} \}$. Assume that the system has the nominal behavior from the beginning, i.e., legitimate agents assign non-zero weights to their legitimate neighbors only and there is no projection of the $s_i[k]$ variables for all $k\geq0$. Then, 
    \begin{align*}
        \sum_{i=1}^{n_{\mL}} \norm{s_i[k+1]} &\leq n_{\mL}G_k + \sum_{i=1}^{n_{\mL}} \norm{s_i[k]} ,
    \end{align*}
    for all $k\geq 0.$
    Moreover, if $s_i[0]=0$ for all $i \in \mL$, we have $$\norm{s_i[k]} \leq \frac{n_{\mL}A}{\exp(\theta)-1}\exp(\theta k),$$ for all $k\geq 0$ and for all $i \in \mL.$
    \label{lemma:s_growth_unbounded}
\end{lemma}
\begin{proof}
    First, we bound the gradients at time $k$. Using the $L$-smoothness of the cost functions and the boundedness of $\Bar{\mX}_k$, we get
    \begin{align*}
        \norm{\nabla f_i(x_i[k])}= \norm{\nabla f_i(x_i[k]) + \nabla f_i(0) - \nabla f_i(0)} \leq L\exp(\theta k)+\norm{\nabla f_i(0)},
    \end{align*}
    where we used the fact that gradients are $L$-Lipschitz continuous. If we define $G_k \triangleq A \exp(\theta k)$ with $A = L+ \max_i \{\norm{\nabla f_i(0)} \}$, we have $\norm{\nabla f_i(x_i[k])} \leq G_k$ for all $i$ and for all $k\geq 0$. Following the same steps in the proof of \Cref{proposition:s_bound} and using the bound on $\norm{\nabla f_i(x_i[k])}$, we obtain
    \begin{align*}
        \sum_{i=1}^{n_{\mL}} \norm{s_i[k+1]} &\leq n_{\mL} G_k +\sum_{i=1}^{n_{\mL}} \norm{s_i[k]}.
    \end{align*}
    Recall that $G_k \triangleq A \exp(\theta k)$. Since $\sum_{i=1}^{n_{\mL}} \norm{s_i[0]}=0,$ by induction, we get 
    \begin{align*}
        \sum_{i=1}^{n_{\mL}} \norm{s_i[k]} \leq n_{\mL} \sum_{t=0}^{k-1} G_t 
        \leq \frac{n_{\mL}A}{\exp(\theta)-1}\exp(\theta k),
    \end{align*}
    where in the last step, we used the definition of $G_k$ and summed the geometric series. Since $\norm{s_i[k]}\leq \sum_{i=1}^{n_{\mL}} \norm{s_i[k]},$ this bound yields the desired relation.
\end{proof}
Notice that the conditions on the growth of the sequence $\{\mX_k\}$ is different from that of $\{\mS_k\}$. This is because the growth rate of $\{\mS_k\}$ is affected by the bounds on the gradients. Therefore, we have more flexibility in choosing the growth rate of the set sequence $\{\mX_k\}$.

\subsection{Convergence Results}\label{sec:convergence_unbounded}
In this section, we derive convergence results for unbounded optimization problems that are analogous to the convergence results in \Cref{sec:asymptotic_results_bounded} and \Cref{sec:finite_time_results_bounded}. First, we give the growing set sequence $\{\mS_k\}$ that helps us achieve the nominal behavior after some point.
\begin{proposition}
    For all $k\geq 0$, define $\Bar{\mX}_k = \mX \cap \mX_k,$ where $\mX_k = \{x \in \R^d \mid \norm{x} \leq \exp (\theta_1 k)\}$ with $\theta_1 >0$. Similarly, define $\mS_k = \{s \in \R^d \mid \norm{s} \leq \exp(\theta_2 k)\}$ with $\theta_2>\theta_1.$ Let $s_i[0]=0$ for all $i \in \mL$. 
    Define 
    \begin{equation}
        T_{nom} \triangleq \max \left\{T_{max}+ \frac{\ln(2n_{\mL})}{\theta_2}, \frac{\ln(A)}{\theta_2-\theta_1},\frac{\ln \norm{x^*}}{\theta_1} \right\},
        \label{eq:tnom}
    \end{equation}
    where $A = \frac{2n_{\mL}(L+\max_i \{\norm{\nabla f_i(0)} \})^2}{\exp(\theta_1)-1}.$ Then, the RP3 algorithm has the nominal behavior for all $k \geq T_{nom}$ almost surely. 
\label{proposition:nominal_behav_unbounded}
\end{proposition}
\begin{proof}
    Define $G_k \triangleq A_1 \exp(\theta_1 k),$ where $A_1 \triangleq L+\max_i \{\norm{\nabla f_i(0)} \}$ and also define
        $A_2 \triangleq G_{T_{max}} \frac{n_{\mL} A_1}{\exp(\theta_1)-1}.$
    Following the proof of \Cref{proposition:removing_projection}, for $k\geq \tmax$ we have 
    \begin{align*}
        \sum_{i\in \mL} \norm{s_i[k]} &\leq \sum_{i\in \mL} \norm{s_i[\tmax]}+ n_{\mL} \sum_{t=\tmax}^{k-1} G_t 
        \leq n_{\mL} \exp(\theta_2 \tmax) + A_2 \exp (\theta_1 (k-T_{max})),
    \end{align*}
    where in the last step, we summed the geometric series. We want to find the time when 
    \begin{equation}
        \exp(\theta_2 k)> n_{\mL} \exp(\theta_2 T_{max}) + A_2 \exp (\theta_1 (k-T_{max}))
        \label{eq:t_nom_unbounded_ineq}
    \end{equation} is satisfied. We deal with the terms on the right hand side of \eqref{eq:t_nom_unbounded_ineq} separately by splitting $\exp(\theta_2 k)$ into two equal terms. For the first one we write
    \begin{align*}
        \frac{1}{2}\exp(\theta_2 k) &> n_{\mL} \exp(\theta_2 T_{max}),\\
        k &> T_{max} + \frac{\ln(2n_{\mL})}{2}.
    \end{align*}
    For the second term, we have
    \begin{align*}
        \frac{1}{2}\exp(\theta_2 k) &> A_2 \exp (\theta_1 (k-T_{max})) \\
        \exp(\theta_2 k) &> 2G_{T_{max}} \frac{n_{\mL} A_1}{\exp(\theta_1)-1}\exp (\theta_1 (k-T_{max})) \\
        \exp(\theta_2 k) &> A \exp (\theta_1 T_{max}) \exp (\theta_1 (k-T_{max})) \\
        k &> \frac{\ln(A)}{\theta_2-\theta_1},
    \end{align*}
    where $A = \frac{2n_{\mL}(L+\max_i \{\norm{\nabla f_i(x')} \})^2}{\exp(\theta_1)-1}$. Then, for any $k\geq \max \left\{T_{max}+ \frac{\ln(2n_{\mL})}{\theta_2}, \frac{\ln(A)}{\theta_2-\theta_1}  \right\},$ the inequality \eqref{eq:t_nom_unbounded_ineq} will be satisfied. Hence, for any $k \geq \max \left\{T_{max}+ \frac{\ln(2n_{\mL})}{\theta_2}, \frac{\ln(A)}{\theta_2-\theta_1}  \right\}$ the projection operator will coincide with the identity operator. Lastly, we must consider the time when $x^*$ is included in the set $\mX_k$. Therefore, once we choose $T_{nom} \triangleq \max \left\{T_{max}+ \frac{\ln(2n_{\mL})}{\theta_2}, \frac{\ln(A)}{\theta_2-\theta_1},\frac{\ln \norm{x^*}}{\theta_1}   \right\}$, the system reaches the nominal behavior after $T_{nom}.$
\end{proof}
Since the system reaches the nominal behavior, we retrieve the almost sure convergence result.
\begin{corollary}
    Let $\Bar{\mX}_k = \mX \cap \mX_k,$ where $\mX_k = \{x \in \R^d \mid \norm{x} \leq \exp (\theta_1 k)\}$ with $\theta_1 >0$ and $\mS_k = \{s \in \R^d \mid \norm{s} \leq \exp(\theta_2 k)\}$ with $\theta_2>\theta_1.$ Then, \Cref{theorem:as_convergence} holds.
    \label{cor:as_convergence_unbounded}
\end{corollary}

The convergence in mean and the convergence in expectation results are more tricky to get in this case. This is due to the potential for malicious agents to exponentially amplify their impact on the system over time, as the bounds on the variables themselves expand exponentially. Therefore, we need to ensure that the learning rate of the agents is faster. First, we show the following bound on $\normphi{\boldx[k] - \boldx^*}$ which will be useful for proving the convergence in the $r$-th mean sense.
\begin{lemma}
    Assume that agents construct the sets $\mX_k = \{x \in \R^d \mid \norm{x} \leq \exp (\theta_1 k)\}$ with $\theta_1 >0$ and $\mS_k = \{s \in \R^d \mid \norm{s} \leq \exp(\theta_2 k)\}$ with $\theta_2>\theta_1.$ Then, for all $k\geq0$ and for a given $r$ with $r\geq 1$, there exists almost surely a constant $c>0$ that depends on $r$ such that
    \begin{equation}
        \normphi{\boldx[k] - \boldx^*}^r \leq c\exp(r \theta_2 \tnom),
    \end{equation}
    where $\tnom$ is as defined in \eqref{eq:tnom}.
    \label{lemma:error_bound_unbounded}
\end{lemma}
\begin{proof}
    Define $c_1$ such that $[M(\eta, \lambda)^k]_{ij} \leq \frac{\min (\pi)}{2(n_{\mL}+1)}$ for all elements of the matrix $i, j$ and $k\geq c_1$. Such a $c_1$ always exists since $\rho(M(\eta, \lambda))<1.$
    We divide the proof in two parts. First, consider the case where $k> \tnom + c_1$. Then, by \Cref{proposition:nominal_behav_unbounded} we have 
    \begin{align*}
        \normphi{\boldx[k] - \boldx^*} &\leq [M(\eta, \lambda)^{k-\tnom}]_{11} 2\exp(\theta_1 \tnom) \\
        &+ [M(\eta, \lambda)^{k-\tnom}]_{12} 2\exp(\theta_1 \tnom) \\
        &+ [M(\eta, \lambda)^{k-\tnom}]_{13} \frac{2(n_{\mL}+1)}{\min (\pi)} \exp(\theta_2 \tnom) \\
        &\leq 2\exp(\theta_1 \tnom) + \exp(\theta_2 \tnom) \\
        &\leq 3\exp(\theta_2 \tnom).
    \end{align*}
    Therefore, we have 
    \begin{equation}
        \normphi{\boldx[k] - \boldx^*}^r \leq 3^r\exp(r \theta_2 \tnom).
        \label{eq:bound_k_geq_tnom}
    \end{equation}
    Next, consider the case where $k\leq \tnom + c_1$ Then, we have 
    \begin{align*}
        \normphi{\boldx[k] - \boldx^*} = \sqrt{\sum_{i\in \mL} \phi_i \normsq{x_i[k]-x^*}} 
        &\overset{(a)}{\leq} \exp (\theta_1 k) + \norm{x^*}\\
        &\leq (\norm{x^*}+1)\exp (\theta_1 k),
    \end{align*}
    where in $(a)$ we used the bounds on $x_i[k]$ and the stochasticity of $\phi$. Using this bound, we get $\normphi{\boldx[k] - \boldx^*}^r \leq (\norm{x^*}+1)^r \exp (r \theta_1 k) \leq (\norm{x^*}+1)^r \exp (r \theta_2 k)$. By choosing $c = \max \{3^r, (\norm{x^*}+1)^r\}$, we obtain the desired equality.
\end{proof}
In the next result, we show that carefully chosen $\theta_1$ and $\theta_2$ can guarantee convergence in the $r$-th mean.
\begin{proposition}
    Let $\Bar{\mX}_k = \mX \cap \mX_k,$ where $\mX_k = \{x \in \R^d \mid \norm{x} \leq \exp (\theta_1 k)\}$ with $\theta_1 >0$ and $\mS_k = \{s \in \R^d \mid \norm{s} \leq \exp(\theta_2 k)\}$. Let $r \geq 1$ be given. Choose $\theta_1$ and $\theta_2$ such that $r \theta_2 < \min \{ 2E_{\mL}^2, 2E_{\mM}^2\}$ and $\theta_1<\theta_2.$ Then, \Cref{theorem:convergence_in_mean} holds true for this given $r$.
    \label{proposition:convergence_in_mean_unbounded}
\end{proposition}
\begin{proof}
    We prove this proposition by utilizing the almost sure convergence we established in \Cref{cor:as_convergence_unbounded} via the Dominated Convergence Theorem \cite{nlar2011}.
    By \Cref{lemma:error_bound_unbounded}, we have 
    \begin{equation*}
        \normphi{\boldx[k] - \boldx^*}^r \leq c\exp(r \theta_2 \tnom),
    \end{equation*}
    for some $c>0.$ Then, if we show that $c\exp(r \theta_2 \tnom)$ has a finite expectation, i.e, $\E[c\exp(r \theta_2 \tnom)]<\infty$, the desired result is obtained. Observe that the constant $c$ in \Cref{lemma:error_bound_unbounded} does not depend on $\tnom$. By the law of total expectation we have 
    \begin{align*}
        \E[c\exp(r \theta_2 \tnom)] = \E[\E[c\exp(r \theta_2 \tnom)|\tnom]]
        &= \sum_{k=0}^{\infty} \E[c\exp(r \theta_2 \tnom)|\tnom=k] \Pr(\tnom=k) \\
        &\leq \sum_{k=0}^{\infty} c\exp(r \theta_2 k)\Pr(\tnom=k).
    \end{align*}
    By definition of $\tnom$ given in \eqref{eq:tnom}, $\Pr(\tnom=k) = \Pr(\tmax=k-\frac{\ln(2n_{\mL})}{\theta_2})$ for all $k\geq k'$ for some $k'$. Then, we can use the bounds provided in \Cref{proposition:tmax_prob} to obtain
    \begin{align*}
        \Pr(T_{max}=k) &\leq \min \{p_c(k-\Delta),1\},
    \end{align*}
    where $p_c(k) = N_{\mL}\exp(-2kE_{\mL}^2)+N_{\mL}\exp(-2kE_{\mM}^2)$ to bound the probability $\Pr(\tnom=k).$ Therefore, for all $k\geq k'$, we have $\Pr(\tnom=k) \leq c_1 (\exp(-2kE_{\mL}^2) + \exp(-2kE_{\mM}^2))$ for some constant $c_1>0.$
    
    Using this bound we have 
    \begin{align*}
    \sum_{k=0}^{\infty} c\exp(r \theta_2 k)\Pr(\tnom=k)
        &=\sum_{k=0}^{k'} c\exp(r \theta_2 k)\Pr(\tnom=k) + \sum_{k=k'+1}^{\infty} c\exp(r \theta_2 k)\Pr(\tnom=k) \\
        &\leq \sum_{k=0}^{k'} c\exp(r \theta_2 k)\Pr(\tnom=k)  +\sum_{k=k'+1}^{\infty}c\exp(r \theta_2 k)c_1 (\exp(-2kE_{\mL}^2) + \exp(-2kE_{\mM}^2)).
    \end{align*}
    The first term is a finite sum and thus bounded. The second term is equal to 
    \begin{align*}
        c\cdot c_1 \sum_{k=k'+1}^{\infty} \exp((r \theta_2 - 2E_{\mL}^2) k) + \exp((r \theta_2 -2 E_{\mM}^2)k).
    \end{align*}
    If we choose $r \theta_2 < \min \{ 2E_{\mL}^2, 2E_{\mM}^2\},$ the second term is also finite, which concludes the proof.
\end{proof}

Finally, we retrieve the expected convergence rate result of \Cref{theorem:exp_conv_rate}. We show that with appropriately chosen growth rates of the sets $\{\mX_k\}$ and $\{\mS_k\}$, our algorithm attains geometric convergence rate in expectation. 

\begin{theorem}
    Let $\Bar{\mX}_k = \mX \cap \mX_k,$ where $\mX_k = \{x \in \R^d \mid \norm{x} \leq \exp (\theta_1 k)\}$ with $\theta_1 >0$ and $\mS_k = \{s \in \R^d \mid \norm{s} \leq \exp(\theta_2 k)\}$ with $\theta_2>\theta_1.$ Define $\Delta_T \triangleq \Delta + \frac{\ln(A)}{\theta_2-\theta_1},$ where $A = \frac{2n_{\mL}(L+\max_i \{\norm{\nabla f_i(0)} \})^2}{\exp(\theta_1)-1}$ and $\Delta$ is as defined in \Cref{proposition:tmax_prob}.
    Then, for all $k\geq \frac{\ln \norm{x^*}}{\theta_1}$, we have
    \begin{equation}
        \begin{aligned}
        \E[\bolde[k]] &\leq
         M(\eta, \lambda)^{k-\lfloor k/2 \rfloor} (I-\!\!M(\eta, \lambda))^{-1}\!\!  
        \begin{bmatrix}
            2\exp(\theta_1 \lfloor k/2 \rfloor) \\ 2\exp(\theta_1 \lfloor k/2 \rfloor) \\ \frac{2(n_{\mL}+1)}{\min (\pi)} \exp(\theta_2 \lfloor k/2 \rfloor\!)\!
        \end{bmatrix} \\ &+\min \{p_e(\lfloor k/2 \rfloor+1-\Delta_T),1\} \begin{bmatrix}
            2\exp(\theta_1 k) \\ 2\exp(\theta_1 k) \\ \frac{2(n_{\mL}+1)}{\min (\pi)} \exp(\theta_2 k)
        \end{bmatrix}.
    \end{aligned}
    \label{eq:exp_conv_rate_unbounded}
    \end{equation}
    Moreover, if we have $$\theta_2 < \min \{-\ln(\rho(M(\eta, \lambda))), E_{\mL}^2,  E_{\mM}^2 \},$$ then the expected convergence is exponential in $k$ for a sufficiently large $k$.
    \label{theorem:expected_conv_rate_unbounded}
\end{theorem}
\begin{proof}
    First, notice that $T_{nom}-T_{max}$ can be at most $\frac{\ln(A)}{\theta_2-\theta_1}$ from the definition of $T_{nom}.$ Therefore, we can use the inequality $\Pr(T_{nom}>k-1) \leq \min \{p_e(k+1-\Delta_T),1\}$ which follows from \Cref{proposition:tmax_prob}. \eqref{eq:exp_conv_rate_unbounded} can be derived using the same steps as the proof of \Cref{theorem:exp_conv_rate}. To show the exponential decrease, we need to show that both terms on the right-hand side of \eqref{eq:exp_conv_rate_unbounded} are exponentially decreasing with $k.$ For the first term, we have
    \begin{align*}
        M(\eta, \lambda)^{k-\lfloor k/2 \rfloor} (I-M(\eta, \lambda))^{-1} \begin{bmatrix}
            2\exp(\theta_1 \lfloor k/2 \rfloor) \\ 2\exp(\theta_1 \lfloor k/2 \rfloor) \\ \frac{2(n_{\mL}+1)}{\min (\pi)} \exp(\theta_2 \lfloor k/2 \rfloor)
        \end{bmatrix}  
        &\leq \exp(\theta_2 \lfloor k/2 \rfloor)M(\eta, \lambda)^{k-\lfloor k/2 \rfloor} v, \\
        &\leq (\exp(\theta_2)M(\eta, \lambda))^{k-\lfloor k/2 \rfloor} v
    \end{align*}
    where $ v= (I-M(\eta, \lambda))^{-1} \begin{bmatrix}
            2 \\
            2 \\ 
            \frac{2(n_{\mL}+1)}{\min (\pi)}
        \end{bmatrix}$. Here, we used the fact that $\theta_2 \geq \theta_1.$ Hence, guaranteeing that $\rho(\exp(\theta_2)M(\eta, \lambda))<1$ guarantees the exponential decrease for this term. For the second term, we have
    \begin{align*}
        \min \{p_e(\lfloor k/2 \rfloor+1-\Delta_T),1\} \begin{bmatrix}
            2\exp(\theta_1 k) \\ 2\exp(\theta_1 k) \\ \frac{2(n+1)}{\min (\pi)} \exp(\theta_2 k)
        \end{bmatrix} 
        &\leq \min \{p_e(k/2-\Delta_T),1\} \exp(\theta_2 k)\begin{bmatrix}
            2\\ 2 \\ \frac{2(n_{\mL}+1)}{\min (\pi)}
        \end{bmatrix}.
    \end{align*}
    For a sufficiently large $k,$ we get an exponential decrease from $p_e(k/2-\Delta_T).$ Hence, we only need to guarantee that the exponential decrease rate we get from $p_e(k/2-\Delta_T)$ is faster than the increase rate of $\exp(\theta_2 k).$ By using the definition of $p_e(k)$ given in \eqref{eq:p_e}, we get the conditions $\theta_2<E_{\mL}^2$ and $\theta_2 < \E_{\mM}^2.$ Combining all three conditions gives us the desired result.
\end{proof}
Different from the analogous result with the compact constraint set $\mX$ in \Cref{theorem:exp_conv_rate}, Theorem~\ref{theorem:expected_conv_rate_unbounded} shows that we have more control over the errors coming from the malicious agents by choosing an appropriate growth rate for the set sequence $\{ \mX_k \}$.
\begin{remark}
    The conditions on the growth rates $\theta_1$ and $\theta_2$ in \Cref{theorem:expected_conv_rate_unbounded} are more strict than the one given in \cref{proposition:convergence_in_mean_unbounded}. This is because while \cref{proposition:convergence_in_mean_unbounded} ensures convergence in expectation, \Cref{theorem:expected_conv_rate_unbounded} provides sufficient conditions for \emph{geometric} convergence rate. 
\end{remark}

\subsection{Choosing Sets $\mX_k$ and $\mS_k$}

Our results in \Cref{theorem:expected_conv_rate_unbounded} require legitimate agents to know the expectations of trust observations $\E_{\mL}$ and $\E_{\mM}$, and \Cref{proposition:convergence_in_mean_unbounded} further requires knowledge of the spectral radius of the matrix $M(\eta, \lambda)$ or at least meaningful upper bounds on these terms.
Some of these system parameters such as $M(\eta, \lambda)$, can be difficult to estimate. This requirement comes from the fact that the sets $\mX_k$ and $\mS_k$ both grow exponentially. Their exponential growth allows for malicious agents to have exponentially growing impact on the system until they are fully detected and excluded. Therefore, the learning part where malicious agents are detected needs to happen faster than the growth of their impact. On the other hand, choosing sets that grow faster decreases the time that the system reaches the nominal behaviour. Hence, there is a trade-off in the choice of the growth rate of these sets. In any case, one may want to guarantee geometric convergence rate without knowing all the system parameters. This is indeed possible since our convergence results depend on two key properties of these sets: 1)~the growth of $\mX_k$ ensures that $x^*$ is eventually included in $\mX_k$, 2)~$\mS_k$ grows faster than $\mX_k$. The first one is required to ensure that $x^*$ is eventually included in $\mX_k$. The second one is required since the norm of the gradient tracking variables $s_i[k]$ grows with a rate that depends on the norm of the gradients, which is related to the bounds on the sets $\mX_k$. This relationship can be seen in \Cref{proposition:s_growth_compact} and \Cref{lemma:s_growth_unbounded}. Therefore, as long as we can satisfy these two conditions, the system will reach the nominal behavior and the convergence will occur. We summarize this discussion formally in the next proposition.

\begin{proposition}
    Assume that we have $\mX_k = \{x \in \R^d \mid \norm{x} \leq g(k)\}$ and $\mS_k = \{s \in \R^d \mid \norm{s} \leq h(k)\}$. Then, if $g(k)$ is an increasing function and $h(k)$ grows faster than $kg(k)$, i.e., $h(k)=\omega(kg(k))$ where $\omega(\cdot)$ denotes the asymptotic lower bound, the nominal behavior is reached almost surely.
    \label{proposition:choosing_bounds}
\end{proposition}
\begin{proof}
    The proof follows similar to \Cref{lemma:s_growth_unbounded}. By \Cref{assumption:cost_function} we have 
    \begin{align*}
        \norm{\nabla f_i(x_i[k])} &= \norm{\nabla f_i(x_i[k]) + \nabla f_i(0) - \nabla f_i(0)} \\
        &\leq Lg(k)+\norm{\nabla f_i(0)}.
    \end{align*}
    Using similar steps to \Cref{lemma:s_growth_unbounded}, we get
    \begin{align*}
        \sum_{i=1}^{n_{\mL}} \norm{s_i[k+1]} &\leq n_{\mL} (Lg(k)+\max_{i \in \mL} \norm{\nabla f_i(0)}) \\ &+\sum_{i=1}^{n_{\mL}} \norm{s_i[k]},
    \end{align*}
    and for $k\geq \tmax$ we have 
    \begin{align*}
        \sum_{i\in \mL} \norm{s_i[k]} &\leq \sum_{i\in \mL} \norm{s_i[\tmax]}+ n_{\mL} \sum_{t=\tmax}^{k-1} Lg(t)
        +\max_{i \in \mL} \norm{\nabla f_i(0)} \\
        &\leq n_{\mL} h(\tmax) + 
        n_{\mL}(k-\tmax)(Lg_x(k-1)
        +\max_{i \in \mL} \norm{\nabla f_i(0)}).
    \end{align*}
    To reach the nominal convergence, we need to have a $k'$ such that for all $k\geq k'$ we have
    \begin{align*}
        h(k)> n_{\mL} h(\tmax) + 
        n_{\mL}(k-\tmax)(Lg(k-1)\\
        +\max_{i \in \mL} \norm{\nabla f_i(0)}).
    \end{align*}
    Such a $k'$ always exists if we choose $h(k)$ to grow faster than $kg(k).$ Hence, the nominal behavior can be reached and the algorithm converges almost surely.
\end{proof}
\begin{remark}
    Legitimate agents can choose a growing set $\mX_k$ and a corresponding $\mS_k$ with a sub-exponential growth rate such that convergence almost surely occurs without knowing any of the system parameters.
\end{remark}
Here, the correct choice of $\{\mX_k\}$ and a corresponding $\{\mS_k\}$ needs to ensure that the nominal behavior is always reached. Recall that the misclassification probabilities of the malicious agents decrease exponentially as shown in \Cref{proposition:tmax_prob}. Therefore, if we $\{\mX_k\}$ and $\{\mS_k\}$ with a sub-exponential rate, the learning part will always occur faster than the growth of sets and we guarantee convergence almost surely.

\subsection{Average Consensus Problem}
In this part, we show that the RP3 algorithm can be used to solve the average consensus problem in the presence of malicious agents without introducing any deviation from the optimal solution.
\begin{corollary}[Average Consensus]
    Consider a multi-agent system where each legitimate agent has an initial value $x_i[0]\in \R$. The goal of legitimate agents is to compute the average of their initial values denoted by $\hat{x}$, i.e, to find
    \begin{equation}
        \hat{x} \triangleq \frac{1}{n_\mL}\sum_{i\in \mL} x_i[0].\label{eq:consensus_defn}
    \end{equation}
    Let \Cref{assumption:trust-observations} and \Cref{assumption:graph-connectivity} hold true. Then, agents can use the RP3 algorithm to compute $\hat{x}$ almost surely, even in the presence of malicious agents.
    \label{corollary:avg_consensus}
\end{corollary}
\begin{proof}
    The average consensus problem can be formulated as
    \begin{align*}
        \min_{x\in \R}\, \frac{1}{n_\mL}\sum_{i\in \mL} (x - x_i[0])^2.
    \end{align*}
    Notice that the optimum solution $x^*$ of this minimization problem is $x^* = \frac{1}{n_\mL}\sum_{i\in \mL} x_i[0].$ Moreover, the local cost functions $f_i(x) \triangleq (x - x_i[0])^2$ is $L-$ smooth and strongly convex. Therefore, by \cref{cor:as_convergence_unbounded}, agents can obtain the optimal solution almost surely using the RP3 algorithm.
\end{proof}
Since we can formulate the consensus problem as a distributed optimization problem with $L-$smooth and strongly convex cost functions, all results obtained in \cref{sec:unbounded_case} are applicable if agents employ the RP3 algorithm to solve the problem. These results improve upon the previous results on the average consensus in the presence of malicious agents when stochastic trust observations are available. The study in \cite{ourTRO} tackles this problem in undirected graphs but it requires values of the legitimate and malicious agents to be bounded by a constant $\eta$. Moreover, the proposed method cannot remove the impact of the malicious agents completely, resulting in a deviation from the true consensus value $\hat{x}$ unlike our results. The resilient distributed optimization algorithm developed in \cite{yemini2022resilience} ensures convergence to the true consensus value for constrained consensus problems, but only for undirected communication graphs and when the mixing matrix is doubly stochastic. Furthermore, its diminishing step sizes lead to a significantly slower convergence rate compared to RP3. The algorithm proposed in \cite{hadjicostis_2022} removes the effect of malicious agents over time. However, the results only show asymptotic convergence while  here, we also characterize the expected convergence rate. 

\section{Numerical Studies}
In this section, we present our numerical studies to validate our theoretical results on two optimization problems. First, we consider a constrained average consensus problem to compare RP3 against the existing methods that leverage inter-agent trust values. Then, we test our algorithm on multi-robot target tracking problem and compare its performance against a data-based resilient optimization benchmark algorithm. The first setup is a constrained optimization problem over undirected graphs while the second setup is an unconstrained optimization problem over a directed graph, covering different use cases for our algorithm.
\subsection{Constrained Consensus} \label{subsec:constrained_consensus_exp}
Here, the agents' goal is to solve the average consensus problem defined in \cref{eq:consensus_defn} subject to constraints. We compare the performance of our algorithm against several benchmarks. We consider resilient consensus (RC) \cite{ourTRO} and trustworthy distributed average consensus (RDA) \cite{hadjicostis_2022} algorithms that are specifically designed for the consensus problem. We also include the resilient distributed optimization (RDO) algorithm in \cite{yemini2022resilience} and Projected Push-Pull (PPP) algorithm as our benchmarks. PPP algorithm is oblivious to the existence of malicious agents in the system, showcasing the worst-case scenario under malicious attacks.

In the experiments, we consider a setup with $L=50$ legitimate agents and $M=100$ malicious agents. We generate an undirected cyclic graph among the legitimate agents and add additional edges between them randomly. We also add undirected edges between malicious and legitimate agents with a probability $0.7$. This process results in a graph where many legitimate agents have more malicious neighbors than legitimate ones, making data-based methods such as \cite{sundaram2018distributed, Kuwaranancharoen2024ScalableDO, Kaheni2022ResilientCO} inapplicable. We choose $[-50,50]$ as the constraint set and generate random initial points for legitimate agents within this range. If the optimal point is positive, malicious agents send $-50$ and they send $50$ otherwise. Malicious agents also send these values as the gradient tracking variable against RP3 and PPP algorithms. Since these methods do not filter based on the values sent by malicious agents, using these extreme values represents a strong attack against all trust-based methods. Similar to previous work \cite{L4DC,yemini2022resilience,ourTRO}, trust values are sampled from the uniform distribution  $[0.35, 0.75]$ for legitimate neighbors and from $[0.25, 0.65]$ for malicious neighbors. As required by RC, agents update their trust opinions for an initial observation window of $30$ iterations before starting to run the optimization algorithms. We use the same trust values for all algorithms. For RP3, we choose the bounds on sets $\mX_k$ and $\mS_k$ as $0.1k$ and $0.1k^2,$ respectively. The results are shown in \cref{fig:consensus_results}. We observe that RP3 has a significantly faster convergence rate than the other algorithms, even though it is not specifically designed for consensus, unlike \cite{ourTRO,hadjicostis_2022}, and is not limited to undirected graphs and constraint optimization, as in \cite{yemini2022resilience}. 
\begin{figure}
    \centering
    \includegraphics[width=0.45\textwidth]{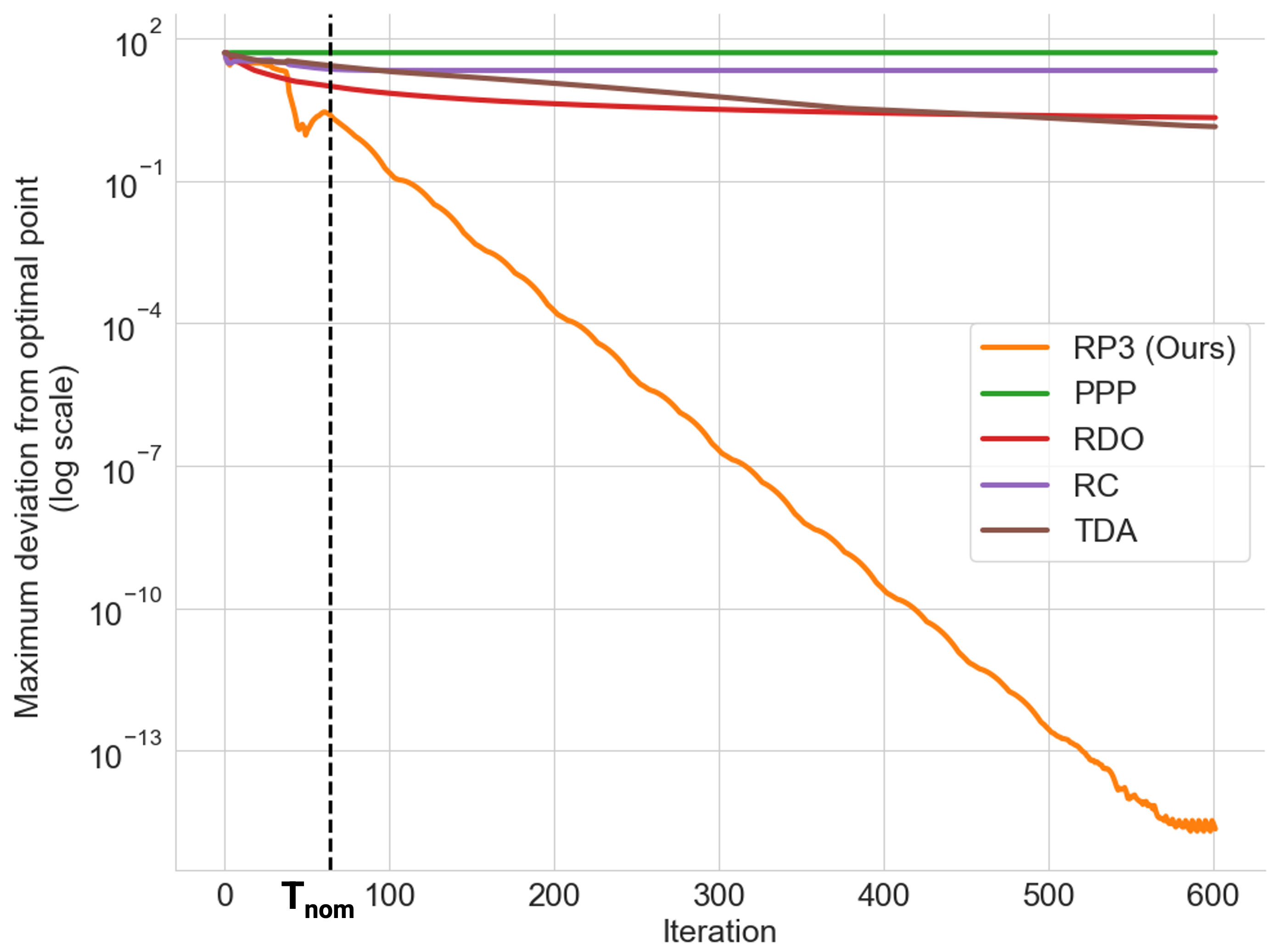}
    \caption{Constrained consensus experiment results with $L=50$ legitimate and $M=100$ malicious agents.}
    \label{fig:consensus_results}
\end{figure}

\subsection{Multi-Robot Target Tracking}
\begin{figure*}[t]
    \centering
    \subfigure[Loss]{\includegraphics[width=0.45\textwidth]{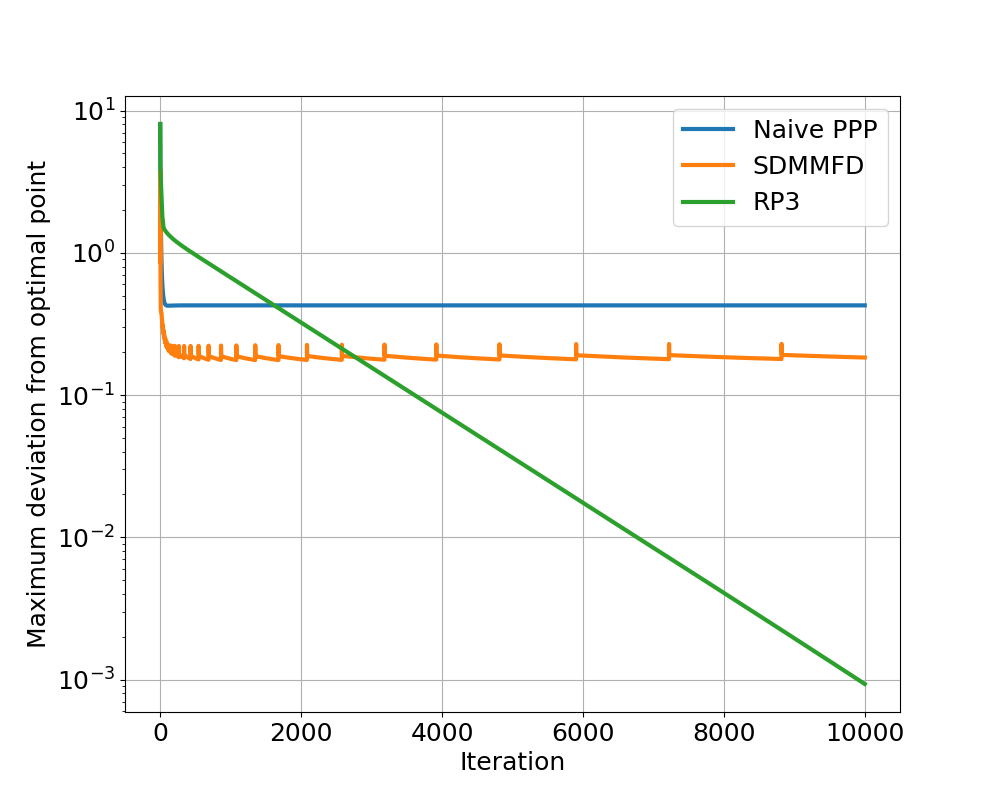}}
    \quad
    \subfigure[Trajectory]{\includegraphics[width=0.45\textwidth]{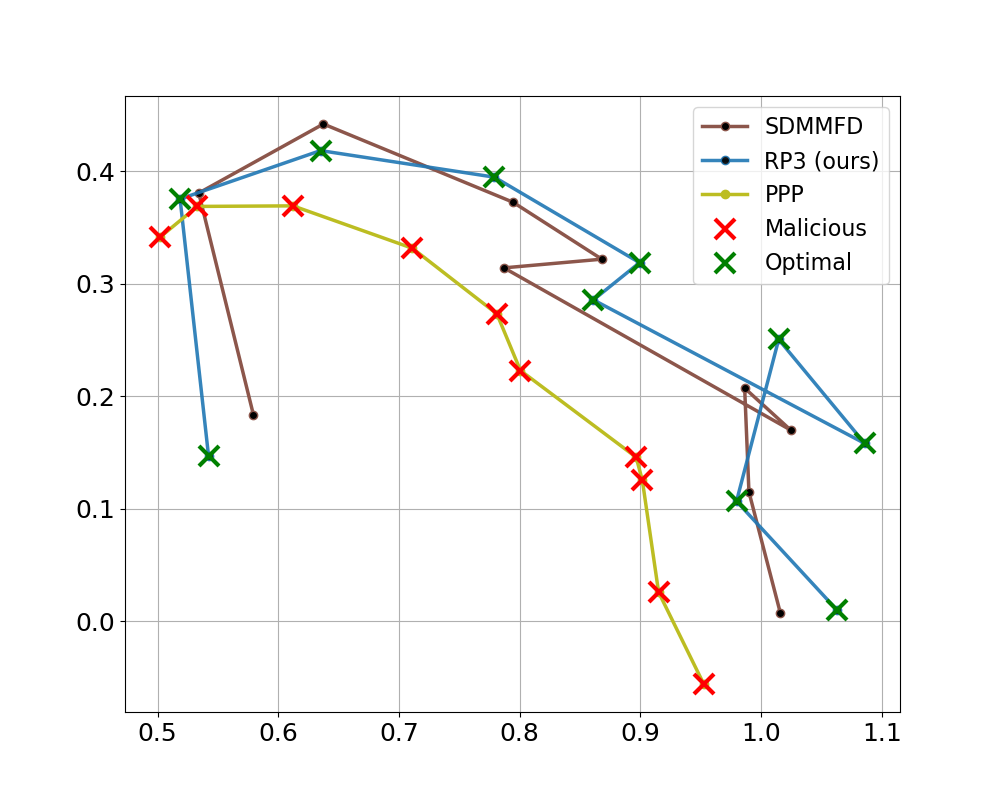}}
    \caption{Target tracking experiment results with $L=9$ legitimate and $M=6$ malicious agents. (a) Convergence to the optimal trajectory. (b) Visualization of the final trajectories for all methods.}
    \label{fig:trajectory_results}
\end{figure*}
In this part, we demonstrate the performance of our algorithm on the distributed multi-robot target tracking problem, as defined in \cite{schwager_tutorial_2024}. In this problem, a set of robots communicating over a directed graph aim to estimate the trajectory of a moving target through local observations and local communication, in the presence of malicious agents. Concretely, the position and velocity at time $t$ of the moving target in a 2D plane is given by a vector $x_t \in \R^4$. The target is assumed to move according to the linear dynamics $x_{t+1} = A_t x_t + w_t$ where $A_t \in \R^{4\times 4}$ captures the known dynamics, and $w_t 
\sim \mN(0, Q_t)$ is the Gaussian process noise. Robot $i$ receives noisy observations of the position of the target if the target is within its observation range, according to $y_{i,t} = x^p_t + v_{i,t}$ where $x^p_t$ denotes the position of the target, and $v_{i,t} \sim \mN(0, R_{i,t})$ is the observation noise. We let $\mT_i$ the set of all times when robot $i$ observes the target. The initial position and velocity $x_0$ is modeled by a prior $\mN(\overline{x_0}, \overline{P_0})$. In this setup, the maximum likelihood estimate is given by minimizing the following loss:
\begin{align}
\label{eq:trajectory_tracking_obj}
f(x)= & \left\|x_0-\bar{x}_0\right\|_{\bar{P}_0^{-1}}^2+\sum_{t=1}^{T-1}\left\|x_{t+1}-A_t x_t\right\|_{Q_t^{-1}}^2 \cr
& +\sum_{i \in \mathcal{V}} \sum_{t \in \mathcal{T}_i}\left\|y_{i, t}-x^p_t\right\|_{R_{i,t}^{-1}}^2.
\end{align}
Then, each robot has the following local cost function which is determined according to their observations
\begin{align*}
f_i(x)= & \frac{1}{N}\left\|x_0-\bar{x}_0\right\|_{\bar{P}_0^{-1}}^2+\sum_{t=1}^{T-1} \frac{1}{N}\left\|x_{t+1}-A_t x_t\right\|_{Q_t^{-1}}^2 \\
& +\sum_{t \in \mathcal{T}_i}\left\|y_{i, t}-C_{i, t} x_t\right\|_{R_{i,t}^{-1}}^2 .
\end{align*}

For this problem, we compare our method against both non-resilient and resilient distributed optimization algorithms. Similar to the constrained consensus experiments, we choose PPP as the non-resilient benchmark. For the resilient benchmark, we implement the SDMMFD algorithm in \cite{Kuwaranancharoen2024ScalableDO}, which is among the best-performing algorithms in its class \cite{Kuwaranancharoen2023OnTG}. This algorithm does not use trust values, and instead relies on outlier elimination to guarantee convergence to the convex hull of the legitimate agents' local minimizers. 

 We consider a setup with $L = 9$, $M = 6$ agents. Legitimate agents are connected through a grid graph with diagonal connections randomly included. The resulting graph is $3$-Robust, satisfying the conditions for the SDMMFD algorithm to converge. Each malicious agent is only connected to $1$ legitimate agent through an undirected edge. Malicious agents send a constant attack value at all time steps, to prevent the legitimate agents from converging to the global minimizer. We consider a time horizon $T=10$ for the trajectory. We maintain the same parameters for the trust observations and the RP3 algorithm as used in the constrained consensus experiments and only adjust the learning rate. The results are shown in \cref{fig:trajectory_results}. We observe that RP3 algorithm converges to the optimal trajectory that minimizes the global objective function in \eqref{eq:trajectory_tracking_obj} (shown in green in \cref{fig:trajectory_results}). The non-resilient Projected Push-Pull converges to the trajectory provided by the malicious agents, and SDMMFD converges to a non-optimal trajectory. Note that SDMMFD and RP3 require different assumptions to guarantee resilience, but we demonstrate that RP3 achieves better performance when trust values are available.

\section{Conclusion}
In this work, we study the distributed optimization problem in the presence of malicious agents, where the objective of legitimate agents is to minimize the sum of their individual strongly convex loss functions by exchanging information over a directed communication graph. Our proposed algorithm, Resilient Projected Push-Pull (RP3), leverages gradient tracking for fast convergence and exploits inter-agent trust values and growing constraint sets to mitigate and eliminate the impact of malicious agents on the system. We characterize the sufficient conditions on the growing sets to ensure geometric convergence in expectation. Our theoretical analysis demonstrates that RP3 converges to the nominal optimal solution almost surely and in the $r$-th mean for any $r\geq 1$, given appropriately chosen step sizes and constraint sets. We show that RP3 can solve average consensus problems faster than existing methods that leverage trust values, including those specifically designed for the consensus problem.

We validate our theoretical results through numerical studies on average consensus and multi-robot target tracking problems. In these studies, our algorithm outperforms both existing trust value-based methods and data-based filtering algorithms. This demonstrates the practical effectiveness of RP3 in enhancing the resilience and efficiency of distributed optimization in adversarial environments.

\bibliographystyle{IEEEtran}
\bibliography{references.bib}

\begin{thebibliography}{10}
\providecommand{\url}[1]{#1}
\csname url@samestyle\endcsname
\providecommand{\newblock}{\relax}
\providecommand{\bibinfo}[2]{#2}
\providecommand{\BIBentrySTDinterwordspacing}{\spaceskip=0pt\relax}
\providecommand{\BIBentryALTinterwordstretchfactor}{4}
\providecommand{\BIBentryALTinterwordspacing}{\spaceskip=\fontdimen2\font plus
\BIBentryALTinterwordstretchfactor\fontdimen3\font minus \fontdimen4\font\relax}
\providecommand{\BIBforeignlanguage}[2]{{%
\expandafter\ifx\csname l@#1\endcsname\relax
\typeout{** WARNING: IEEEtran.bst: No hyphenation pattern has been}%
\typeout{** loaded for the language `#1'. Using the pattern for}%
\typeout{** the default language instead.}%
\else
\language=\csname l@#1\endcsname
\fi
#2}}
\providecommand{\BIBdecl}{\relax}
\BIBdecl

\bibitem{schwager_tutorial_2024}
O.~Shorinwa, T.~Halsted, J.~Yu, and M.~Schwager, ``Distributed optimization methods for multi-robot systems: Part 1—a tutorial,'' \emph{IEEE Robotics \& Automation Magazine}, pp. 2--19, 2024.

\bibitem{nedic_control_2018}
A.~Nedi{\'c} and J.~Liu, ``Distributed optimization for control,'' \emph{Annual Review of Control, Robotics, and Autonomous Systems}, vol.~1, no. Volume 1, 2018, pp. 77--103, 2018.

\bibitem{schwager_collaborative_manipulation}
O.~Shorinwa and M.~Schwager, ``Distributed contact-implicit trajectory optimization for collaborative manipulation,'' in \emph{2021 International Symposium on Multi-Robot and Multi-Agent Systems (MRS)}, 2021, pp. 56--65.

\bibitem{christofides_MPC_2013}
P.~D. Christofides, R.~Scattolini, D.~{Mu{\~n}oz de la Pe{\~n}a}, and J.~Liu, ``Distributed model predictive control: A tutorial review and future research directions,'' \emph{Computers \& Chemical Engineering}, vol.~51, pp. 21--41, 2013, cPC VIII.

\bibitem{roberto_localization_2014}
R.~Tron and R.~Vidal, ``Distributed 3-d localization of camera sensor networks from 2-d image measurements,'' \emph{IEEE Transactions on Automatic Control}, vol.~59, no.~12, pp. 3325--3340, 2014.

\bibitem{Long_localization_2016}
{Dang, Vy-Long}, {Le, Binh-Son}, {Bui, Trong-Tu}, {Huynh, Huu-Thuan}, and {Pham, Cong-Kha}, ``A decentralized localization scheme for swarm robotics based on coordinate geometry and distributed gradient descent,'' \emph{MATEC Web of Conferences}, vol.~54, p. 02002, 2016.

\bibitem{nedic_distributed_sensors}
S.~S. Ram, V.~V. Veeravalli, and A.~Nedi\'c, ``Distributed and recursive parameter estimation in parametrized linear state-space models,'' \emph{IEEE Transactions on Automatic Control}, vol.~55, no.~2, pp. 488--492, 2010.

\bibitem{ola_tracking_2020}
O.~Shorinwa, J.~Yu, T.~Halsted, A.~Koufos, and M.~Schwager, ``Distributed multi-target tracking for autonomous vehicle fleets,'' in \emph{2020 IEEE International Conference on Robotics and Automation (ICRA)}, 2020, pp. 3495--3501.

\bibitem{nedic2014distributed}
A.~Nedi{\'c} and A.~Olshevsky, ``Distributed optimization over time-varying directed graphs,'' \emph{IEEE Transactions on Automatic Control}, vol.~60, no.~3, pp. 601--615, 2014.

\bibitem{makhdoumi2015graph}
A.~Makhdoumi and A.~Ozdaglar, ``Graph balancing for distributed subgradient methods over directed graphs,'' in \emph{2015 54th IEEE Conference on Decision and Control (CDC)}.\hskip 1em plus 0.5em minus 0.4em\relax IEEE, 2015, pp. 1364--1371.

\bibitem{xi2017distributed}
C.~Xi, Q.~Wu, and U.~A. Khan, ``On the distributed optimization over directed networks,'' \emph{Neurocomputing}, vol. 267, pp. 508--515, 2017.

\bibitem{sundaram2010distributed}
S.~Sundaram and C.~N. Hadjicostis, ``Distributed function calculation via linear iterative strategies in the presence of malicious agents,'' \emph{IEEE Transactions on Automatic Control}, vol.~56, no.~7, pp. 1495--1508, 2010.

\bibitem{sundaram2018distributed}
S.~Sundaram and B.~Gharesifard, ``Distributed optimization under adversarial nodes,'' \emph{IEEE Transactions on Automatic Control}, vol.~64, no.~3, pp. 1063--1076, 2018.

\bibitem{ourTRO}
M.~Yemini, A.~Nedi\'{c}, A.~J. Goldsmith, and S.~Gil, ``Characterizing trust and resilience in distributed consensus for cyberphysical systems,'' \emph{IEEE Trans. Robot.}, vol.~38, no.~1, pp. 71--91, 2022.

\bibitem{yemini2022resilience}
M.~Yemini, A.~Nedi{\'c}, S.~Gil, and A.~J. Goldsmith, ``Resilience to malicious activity in distributed optimization for cyberphysical systems,'' in \emph{2022 IEEE 61st Conference on Decision and Control (CDC)}.\hskip 1em plus 0.5em minus 0.4em\relax IEEE, 2022, pp. 4185--4192.

\bibitem{cavorsi_rht_2023}
M.~Cavorsi, O.~E. Akgün, M.~Yemini, A.~J. Goldsmith, and S.~Gil, ``Exploiting trust for resilient hypothesis testing with malicious robots,'' in \emph{2023 IEEE International Conference on Robotics and Automation (ICRA)}, 2023, pp. 7663--7669.

\bibitem{gil2017guaranteeing}
S.~Gil, S.~Kumar, M.~Mazumder, D.~Katabi, and D.~Rus, ``Guaranteeing spoof-resilient multi-robot networks,'' \emph{Autonomous Robots}, vol.~41, no.~6, pp. 1383--1400, 2017.

\bibitem{Pierson2016}
A.~Pierson and M.~Schwager, \emph{Adaptive Inter-Robot Trust for Robust Multi-Robot Sensor Coverage}.\hskip 1em plus 0.5em minus 0.4em\relax Cham: Springer International Publishing, 2016, pp. 167--183.

\bibitem{xiong2023securearray}
J.~Xiong and K.~Jamieson, ``Securearray: improving wifi security with fine-grained physical-layer information,'' in \emph{Proceedings of the 19th Annual International Conference on Mobile Computing \& Networking}, ser. MobiCom '13.\hskip 1em plus 0.5em minus 0.4em\relax New York, NY, USA: Association for Computing Machinery, 2013, p. 441–452.

\bibitem{gil2023physicality}
S.~Gil, M.~Yemini, A.~Chorti, A.~Nedić, H.~V. Poor, and A.~J. Goldsmith, ``How physicality enables trust: A new era of trust-centered cyberphysical systems,'' 2023.

\bibitem{aydın2024multiagent}
S.~Aydın, O.~E. Akgün, S.~Gil, and A.~Nedić, ``Multi-agent resilient consensus under intermittent faulty and malicious transmissions (extended version),'' 2024.

\bibitem{hadjicostis2023trustworthy}
C.~N. Hadjicostis and A.~D. Dominguez-Garcia, ``Trustworthy distributed average consensus based on locally assessed trust evaluations,'' \emph{arXiv preprint arXiv:2309.00920}, 2023.

\bibitem{ravi_2019}
N.~Ravi, A.~Scaglione, and A.~Nedić, ``A case of distributed optimization in adversarial environment,'' in \emph{ICASSP 2019 - 2019 IEEE International Conference on Acoustics, Speech and Signal Processing (ICASSP)}, 2019, pp. 5252--5256.

\bibitem{AB_usman}
R.~Xin and U.~A. Khan, ``A linear algorithm for optimization over directed graphs with geometric convergence,'' \emph{IEEE Control Systems Letters}, vol.~2, no.~3, pp. 315--320, 2018.

\bibitem{push-pull}
S.~Pu, W.~Shi, J.~Xu, and A.~Nedić, ``Push–pull gradient methods for distributed optimization in networks,'' \emph{IEEE Transactions on Automatic Control}, vol.~66, no.~1, pp. 1--16, 2021.

\bibitem{xin2019frost}
R.~Xin, C.~Xi, and U.~A. Khan, ``Frost—fast row-stochastic optimization with uncoordinated step-sizes,'' \emph{EURASIP Journal on Advances in Signal Processing}, vol. 2019, no.~1, pp. 1--14, 2019.

\bibitem{harnessing_smoothness_journal}
G.~Qu and N.~Li, ``Harnessing smoothness to accelerate distributed optimization,'' \emph{IEEE Transactions on Control of Network Systems}, vol.~5, no.~3, pp. 1245--1260, 2018.

\bibitem{L4DC}
O.~E. Akgun, A.~K. Dayi, S.~Gil, and A.~Nedi\'c, ``Learning trust over directed graphs in multiagent systems,'' in \emph{Proceedings of The 5th Annual Learning for Dynamics and Control Conference}, ser. Proceedings of Machine Learning Research, N.~Matni, M.~Morari, and G.~J. Pappas, Eds., vol. 211.\hskip 1em plus 0.5em minus 0.4em\relax PMLR, 15--16 Jun 2023, pp. 142--154.

\bibitem{L4DC_2023_extended}
O.~E. Akg{\"u}n, A.~K. Day{\i}, S.~Gil, and A.~Nedi{\'c}, ``Learning trust over directed graphs in multiagent systems (extended version),'' \emph{arXiv preprint arXiv:2212.02661}, 2022.

\bibitem{harnessing_smoothness_cdc}
G.~Qu and N.~Li, ``Harnessing smoothness to accelerate distributed optimization,'' in \emph{2016 IEEE 55th Conference on Decision and Control (CDC)}, 2016, pp. 159--166.

\bibitem{xu2015augmented}
J.~Xu, S.~Zhu, Y.~C. Soh, and L.~Xie, ``Augmented distributed gradient methods for multi-agent optimization under uncoordinated constant stepsizes,'' in \emph{2015 54th IEEE Conference on Decision and Control (CDC)}.\hskip 1em plus 0.5em minus 0.4em\relax IEEE, 2015, pp. 2055--2060.

\bibitem{DIGing}
A.~Nedi\'{c}, A.~Olshevsky, and W.~Shi, ``Achieving geometric convergence for distributed optimization over time-varying graphs,'' \emph{SIAM Journal on Optimization}, vol.~27, no.~4, pp. 2597--2633, 2017.

\bibitem{FROST_og}
C.~Xi, V.~S. Mai, R.~Xin, E.~H. Abed, and U.~A. Khan, ``Linear convergence in optimization over directed graphs with row-stochastic matrices,'' \emph{IEEE Transactions on Automatic Control}, vol.~63, no.~10, pp. 3558--3565, 2018.

\bibitem{liu2020discrete}
H.~Liu, W.~Yu, and G.~Chen, ``Discrete-time algorithms for distributed constrained convex optimization with linear convergence rates,'' \emph{IEEE Transactions on Cybernetics}, vol.~52, no.~6, pp. 4874--4885, 2020.

\bibitem{luan2023distributed}
M.~Luan, G.~Wen, H.~Liu, T.~Huang, G.~Chen, and W.~Yu, ``Distributed discrete-time convex optimization with closed convex set constraints: Linearly convergent algorithm design,'' \emph{IEEE Transactions on Cybernetics}, pp. 1--13, 2023.

\bibitem{scutari2019distributed}
G.~Scutari and Y.~Sun, ``Distributed nonconvex constrained optimization over time-varying digraphs,'' \emph{Mathematical Programming}, vol. 176, pp. 497--544, 2019.

\bibitem{projected_push_pull}
O.~E. Akg{\"u}n, A.~K. Day{\i}, S.~Gil, and A.~Nedi{\'c}, ``Projected push-pull for distributed constrained optimization over time-varying directed graphs (extended version),'' \emph{arXiv preprint arXiv:2310.06223}, 2023.

\bibitem{robust2020}
S.~Pu, ``A robust gradient tracking method for distributed optimization over directed networks,'' in \emph{Proc. of the 59th {IEEE} Conference on Decision and Control ({CDC})}, Jeju Island, Republic of Korea, Dec. 2020, pp. 2335--2341.

\bibitem{basar_noise_GT_2023}
Y.~Wang and T.~Başar, ``Gradient-tracking-based distributed optimization with guaranteed optimality under noisy information sharing,'' \emph{IEEE Transactions on Automatic Control}, vol.~68, no.~8, pp. 4796--4811, 2023.

\bibitem{nedic_DP_GT_2024}
Y.~Wang and A.~Nedić, ``Tailoring gradient methods for differentially private distributed optimization,'' \emph{IEEE Transactions on Automatic Control}, vol.~69, no.~2, pp. 872--887, 2024.

\bibitem{Zhaoye_noise_GT_2024}
Z.~Pan, H.~Yang, and H.~Liu, ``Utilizing second-order information in noisy information-sharing environments for distributed optimization,'' in \emph{ICASSP 2024 - 2024 IEEE International Conference on Acoustics, Speech and Signal Processing (ICASSP)}, 2024, pp. 9156--9160.

\bibitem{wenwen_noise_GT_2023}
W.~Wu, S.~Zhu, S.~Liu, and X.~Guan, ``Exact noise-robust distributed gradient-tracking algorithm for constraint-coupled resource allocation problems,'' in \emph{2023 62nd IEEE Conference on Decision and Control (CDC)}, 2023, pp. 7271--7276.

\bibitem{wu_noise_GT_2022}
W.~Wu, S.~Liu, and S.~Zhu, ``Distributed dual gradient tracking for economic dispatch in power systems with noisy information,'' \emph{Electric Power Systems Research}, vol. 211, p. 108298, 2022.

\bibitem{nedic_DSGT_2021}
S.~Pu and A.~Nedi{\'c}, ``Distributed stochastic gradient tracking methods,'' \emph{Mathematical Programming}, vol. 187, no.~1, pp. 409--457, 2021.

\bibitem{xin_CDC_SGT_2019}
R.~Xin, A.~K. Sahu, U.~A. Khan, and S.~Kar, ``Distributed stochastic optimization with gradient tracking over strongly-connected networks,'' in \emph{2019 IEEE 58th Conference on Decision and Control (CDC)}, 2019, pp. 8353--8358.

\bibitem{zhao_DSGT_2024}
S.~Zhao and Y.~Liu, ``Confidence region for distributed stochastic optimization problem via stochastic gradient tracking method,'' \emph{Automatica}, vol. 159, p. 111352, 2024.

\bibitem{Su2021ByzantineResilientMO}
L.~Su and N.~H. Vaidya, ``Byzantine-resilient multiagent optimization,'' \emph{IEEE Transactions on Automatic Control}, vol.~66, pp. 2227--2233, 2021.

\bibitem{Kuwaranancharoen2024ScalableDO}
K.~Kuwaranancharoen, L.~Xin, and S.~Sundaram, ``Scalable distributed optimization of multi-dimensional functions despite byzantine adversaries,'' \emph{IEEE Transactions on Signal and Information Processing over Networks}, vol.~10, pp. 360--375, 2024.

\bibitem{Zhang2024AcceleratedDO}
S.~Zhang, Z.~Liu, G.~Wen, and Y.~wu~Wang, ``Accelerated distributed optimization algorithm with malicious nodes,'' \emph{IEEE Transactions on Network Science and Engineering}, vol.~11, pp. 2238--2248, 2024.

\bibitem{Gupta2021ByzantineFI}
N.~Gupta, T.~T. Doan, and N.~H. Vaidya, ``Byzantine fault-tolerance in decentralized optimization under 2f-redundancy,'' \emph{2021 American Control Conference (ACC)}, pp. 3632--3637, 2021.

\bibitem{Zhu2022ResilientDO}
J.~Zhu, Y.~Lin, A.~Velasquez, and J.~Liu, ``Resilient distributed optimization*,'' \emph{2023 American Control Conference (ACC)}, pp. 1307--1312, 2022.

\bibitem{Gupta2020ResilienceIC}
N.~Gupta and N.~H. Vaidya, ``Resilience in collaborative optimization: Redundant and independent cost functions,'' \emph{ArXiv}, vol. abs/2003.09675, 2020.

\bibitem{Ravi2019DetectionAI}
N.~Ravi and A.~Scaglione, ``Detection and isolation of adversaries in decentralized optimization for non-strongly convex objectives,'' \emph{IFAC-PapersOnLine}, vol.~52, no.~20, pp. 381--386, 2019.

\bibitem{Zhao2020ResilientDO}
C.~Zhao, J.~He, and Q.-G. Wang, ``Resilient distributed optimization algorithm against adversarial attacks,'' \emph{IEEE Transactions on Automatic Control}, vol.~65, pp. 4308--4315, 2020.

\bibitem{Xu2022ATR}
C.~Xu and Q.~Liu, ``A trust‐based resilient consensus algorithm for distributed optimization considering node and edge attacks,'' \emph{International Journal of Robust and Nonlinear Control}, vol.~33, pp. 3517 -- 3534, 2022.

\bibitem{Fu2020ResilientCD}
W.~Fu, Q.~Ma, J.~Qin, and Y.~Kang, ``Resilient consensus‐based distributed optimization under deception attacks,'' \emph{International Journal of Robust and Nonlinear Control}, vol.~31, pp. 1803 -- 1816, 2020.

\bibitem{Kaheni2022ResilientCO}
M.~Kaheni, E.~Usai, and M.~Franceschelli, ``Resilient constrained optimization in multi-agent systems with improved guarantee on approximation bounds,'' \emph{IEEE Control Systems Letters}, vol.~6, pp. 2659--2664, 2022.

\bibitem{cavorsi2023ICRA}
M.~Cavorsi, O.~E. Akgün, M.~Yemini, A.~J. Goldsmith, and S.~Gil, ``Exploiting trust for resilient hypothesis testing with malicious robots,'' in \emph{2023 IEEE International Conference on Robotics and Automation (ICRA)}, 2023, pp. 7663--7669.

\bibitem{cheng2021general}
M.~Cheng, C.~Yin, J.~Zhang, S.~Nazarian, J.~Deshmukh, and P.~Bogdan, ``A general trust framework for multi-agent systems,'' in \emph{Proceedings of the 20th International Conference on Autonomous Agents and MultiAgent Systems}, 2021, pp. 332--340.

\bibitem{pippin2014trust}
C.~Pippin and H.~Christensen, ``Trust modeling in multi-robot patrolling,'' in \emph{2014 IEEE International Conference on Robotics and Automation (ICRA)}.\hskip 1em plus 0.5em minus 0.4em\relax IEEE, 2014, pp. 59--66.

\bibitem{yang2024enhancing}
Z.~Yang and R.~Tron, ``Enhancing security in multi-robot systems through co-observation planning, reachability analysis, and network flow,'' 2024.

\bibitem{Ding_GT_attack_CDC_2020}
T.~Ding, Q.~Xu, S.~Zhu, and X.~Guan, ``A convergence-preserving data integrity attack on distributed optimization using local information,'' in \emph{2020 59th IEEE Conference on Decision and Control (CDC)}, 2020, pp. 3598--3603.

\bibitem{nguyen2022distributed}
D.~T.~A. Nguyen, D.~T. Nguyen, and A.~Nedi{\'c}, ``Distributed nash equilibrium seeking over time-varying directed communication networks,'' \emph{arXiv preprint arXiv:2201.02323}, 2022.

\bibitem{robust-pp}
S.~Pu, ``A robust gradient tracking method for distributed optimization over directed networks,'' in \emph{2020 59th IEEE Conference on Decision and Control (CDC)}, 2020, pp. 2335--2341.

\bibitem{nlar2011}
E.~{\c{C}}{\i}nlar, \emph{Measure and Integration}.\hskip 1em plus 0.5em minus 0.4em\relax New York, NY: Springer New York, 2011, pp. 1--47.

\bibitem{Varga2000}
R.~S. Varga, \emph{Basic Iterative Methods and Comparison Theorems}.\hskip 1em plus 0.5em minus 0.4em\relax Berlin, Heidelberg: Springer Berlin Heidelberg, 2000, pp. 63--110.

\bibitem{hadjicostis_2022}
C.~N. Hadjicostis and A.~D. Domínguez-García, ``Trustworthy distributed average consensus,'' in \emph{2022 IEEE 61st Conference on Decision and Control (CDC)}, 2022, pp. 7403--7408.

\bibitem{Kuwaranancharoen2023OnTG}
K.~Kuwaranancharoen and S.~Sundaram, ``On the geometric convergence of byzantine-resilient distributed optimization algorithms,'' \emph{ArXiv}, vol. abs/2305.10810, 2023.

\end{thebibliography}

\end{document}